\documentclass[11pt]{article}

\usepackage[numbers,sort&compress]{natbib}
\usepackage[margin=1in]{geometry}
\usepackage{amsmath, amssymb, amsthm, dsfont, mathtools}
\usepackage{aliascnt}
\usepackage{xcolor}

\usepackage{hyperref}
\usepackage{authblk}
\usepackage{tocloft}

\newtheoremstyle{myplain}{5pt}{5pt}{\itshape}{0pt}{\bfseries}{}{5pt plus 1pt minus 1pt}{}
\newtheoremstyle{mydefinition}{5pt}{5pt}{\normalfont}{0pt}{\bfseries}{}{5pt plus 1pt minus 1pt}{}

\theoremstyle{myplain}
\newtheorem{thm}{Theorem}[section]
\newaliascnt{cor}{thm}
\newtheorem{cor}[cor]{Corollary}
\aliascntresetthe{cor}
\newaliascnt{lem}{thm}
\newtheorem{lem}[lem]{Lemma}
\aliascntresetthe{lem}
\newaliascnt{prop}{thm}
\newtheorem{prop}[prop]{Proposition}
\aliascntresetthe{prop}
\newtheorem{assum}{Assumption}
\newtheorem*{assum*}{Assumption}

\newtheorem{stabilityassum}{Assumption}

\newtheorem{regularizationassum}{Assumption}

\theoremstyle{mydefinition}
\newaliascnt{defi}{thm}
\newtheorem{defi}[defi]{Definition}
\aliascntresetthe{defi}

\theoremstyle{remark}
\newtheorem{ex}{Example}
\newtheorem{rmk}[ex]{Remark}

\DeclareMathOperator{\C}{\mathbb{C}}
\DeclareMathOperator{\N}{\mathbb{N}}

\DeclareMathOperator{\cB}{\mathcal{B}}
\DeclareMathOperator{\cD}{\mathcal{D}}
\DeclareMathOperator{\cE}{\mathcal{E}}
\DeclareMathOperator{\cH}{\mathcal{H}}
\DeclareMathOperator{\cL}{\mathcal{L}}
\DeclareMathOperator{\cP}{\mathcal{P}}
\DeclareMathOperator{\cT}{\mathcal{T}}
\DeclareMathOperator{\cW}{\mathcal{W}}

\DeclareMathOperator{\dist}{dist}
\DeclareMathOperator{\1}{\mathds{1}}

\newcommand{\braket}[1]{\left\langle #1 \right\rangle}
\newcommand{\ketbra}[2]{\left\vert#1\right\rangle\!\! \left\langle#2\right\vert}
\newcommand{\ket}[1]{\left\vert #1 \right\rangle}
\newcommand{\bra}[1]{\left\langle #1 \right\vert}
\newcommand{\tr}[1][]{\operatorname{tr}%
  \ifx#1\empty \else \left[ #1 \right] \fi}

\definecolor{TealDark}{HTML}{008080}
\definecolor{SlateBlueDark}{HTML}{6A5ACD}

\hypersetup{
    colorlinks,
    linkcolor={TealDark},
    citecolor={SlateBlueDark},
    urlcolor={SlateBlueDark}
}

\providecommand{\doi}[1]{\href{https://doi.org/#1}{\nolinkurl{https://doi.org/#1}}}
\title{\texorpdfstring{Instantaneous Sobolev Regularization \\ for Dissipative Bosonic Dynamics}{Instantaneous Sobolev Regularization for Dissipative Bosonic Dynamics}}
\date{}
\author[1,2]{Pablo Costa Rico \thanks{pablo.costa@tum.de}}
\author[3]{Paul Gondolf}
\author[3,4]{Tim M\"{o}bus\thanks{moebustim@gmail.com}}

\affil[1]{Department of Mathematics, Technical University of Munich, 80333 M\"unchen, Germany}
\affil[2]{Munich Center for Quantum Science and Technology, 80799 M\"unchen, Germany}
\affil[3]{Department of Applied Mathematics and Theoretical Physics, University of Cambridge, Cambridge CB3 0WA, United Kingdom}
\affil[4]{Department of Mathematics, University of T\"ubingen, 72076 T\"ubingen, Germany}

\begin{document}
\hypersetup{pageanchor=false}
\maketitle
\vspace*{-6ex}
\begin{abstract}
    We investigate quantum Markov semigroups on bosonic Fock space and identify a broad class of infinite-dimensional dissipative evolutions that exhibit instantaneous Sobolev regularization. Motivated by stability problems in quantum computation, we show that for certain Lindblad operators that are polynomials of creation and annihilation operators, the resulting dynamics immediately transform any initial state into one with finite expectation in all powers of the number operator. A key application is in the bosonic cat code, where we obtain explicit estimates in the trace norm for the speed of convergence. These estimates sharpen existing perturbative bounds at both short and long times, offering new analytic tools for assessing stability and error suppression in bosonic quantum information processing. For example, we improve the strong exponential convergence of the (shifted) $2$-photon dissipation to its asymptotic channel to the uniform topology.
    For multi-mode systems, a generation theorem in concentrated single-sandwich norms supplies the domain properties required for the regularization argument.
\end{abstract}

\vspace*{-1ex}
\setlength{\cftbeforesecskip}{2pt}
\setlength{\cftbeforesubsecskip}{1pt}
\tableofcontents
\vspace*{\fill}~

\thispagestyle{empty}
\newpage
\addtocounter{page}{-1}
\hypersetup{pageanchor=true}

\section{Introduction}\label{sec:intro}
    Quantum dynamical semigroups provide the mathematical framework for describing the time evolution of open quantum systems under the physically natural Markov (memoryless) assumption. From operator semigroup theory, it is well known that such a quantum Markov semigroup (QMS) is uniquely characterized by its infinitesimal generator~$\cL$, i.e., by the master equation
    \begin{equation}\label{eq:master-equation}
        \frac{\partial}{\partial t}\rho(t) = \cL(\rho(t)), \qquad \rho(0) = \rho_0,
    \end{equation}
    for appropriate initial states~$\rho_0$. 

    Motivated by well-known examples such as the quantum harmonic oscillator ($\cL = -i[a^\dagger a, \cdot]$), the Bose–Hubbard model ($\cL = -i[\sum_{i,j}\lambda_{ij} a_i^\dagger a_j + u_i (a_i^\dagger)^2 a_i^2 + \mu_i a_i^\dagger a_i, \cdot]$), and two-photon dissipation ($\cL = a^2 \cdot (a^\dagger)^2 - \tfrac{1}{2}\{(a^\dagger)^2 a^2, \cdot\}$), all expressed in terms of the bosonic annihilation $a$ and creation $a^\dagger$ operators defined on the Fock space $\cH$, a separable Hilbert space, it becomes evident that the class of unbounded generators $\cL$ is central to physically relevant applications. As Lindblad emphasized,
    \begin{center}
        \textit{[boundedness is] a condition which is not fulfilled in many applications (We may hope that this restriction can be ultimately removed using more powerful mathematics)} \cite[p.~120]{Lindblad.1976}.
    \end{center}
    This remark originates from Lindblad’s seminal paper \cite{Lindblad.1976} on the classification of bounded generators published concurrently with the work of Gorini, Kossakowski, and Sudarshan \cite{Gorini.1976}. Together, these results yield the celebrated Gorini–Kossakowski–Sudarshan–Lindblad (GKSL) theorem, which establishes that any QMS generator acting on the space~$\cT_1(\cH)$ of trace-class operators on a finite-dimensional Hilbert space $\cH$ takes the form
    \begin{equation}\label{eq:GKSL}
        \cL(\rho) = -i[H, \rho] + \sum_{j=1}^{K}\!\left(L_j \rho L_j^\dagger - \tfrac{1}{2}\{L_j^\dagger L_j, \rho\}\right)\,.
    \end{equation}

    For unbounded Hamiltonians and Lindblad operators, such as those described by polynomial functions of bosonic creation and annihilation operators on infinite-dimensional continuous-variable (CV) systems considered in this work, a direct extension of the GKSL form is highly nontrivial. This difficulty is already apparent in the pure birth process, where
    \begin{equation*}
        \cL = (a^\dagger)^2 \cdot a^2 - \tfrac{1}{2}\{a^2 (a^\dagger)^2, \cdot\}\,,
    \end{equation*}
    has an associated minimal semigroup that fails to preserve probability mass and hence is not a QMS \cite[Example~3.3]{Davies.1977}. Consequently, one must impose additional, sufficient conditions to ensure the existence of a well-defined QMS \cite{Chebotarev.1997,Chebotarev.1998,Chebotarev.2003,Fagnola.1999,Siemon.2017,Cipriani.2000,Carbone.2003,Gondolf.2024}. Quantum stochastic calculus provides a related framework for stochastic evolutions \cite{Hudson.1984}.

    Continuous-variable systems provide a natural setting for quantum information processing \cite{Braunstein.2005} and bosonic quantum error correction \cite{Gottesman.2001,Joshi.2021}. In particular, engineered multiphoton dissipation is central to the stabilization and manipulation of cat-code spaces \cite{Mirrahimi.2014,Leghtas.2015,Azouit.2016,Guillaud.2023}. Experimental demonstrations \cite{Ofek.2016,Rosenblum.2018}, fault-tolerant architectures \cite{Chamberland.2022}, and multi-mode constructions \cite{Jain.2024} further motivate quantitative control of unbounded bosonic dynamics. Recent analytic work has established well-posedness and energy estimates for bosonic quantum Markov semigroups \cite{Agredo.2021,Gondolf.2024}.
    
    In the recent work \cite{Gondolf.2024} (as well as the doctoral theses \cite{Moebus.2025thesis,Gondolf.2025thesis}) two of the present authors have introduced a sufficient condition on the operator in \eqref{eq:GKSL} to guarantee a unique solution of the initial-value problem \eqref{eq:master-equation}. Here $N=a^\dagger a$ denotes the number operator and $\1$ the identity operator. This condition, we term the \emph{Sobolev stability condition}
    \begin{equation}\label{eq:assum-sobolev-stability}
        \tr\bigl[\cL(\rho)(N + \1)^{k}\bigr] \leq \omega_{k} \tr\bigl[\rho(N + \1)^{k}\bigr]
    \end{equation}
    For fixed $k\geq0$, it suffices that \eqref{eq:assum-sobolev-stability} hold at $k$ and at one order $l>k+2d$, where $d$ bounds the polynomial degrees. This gives generation and $k$-Sobolev preservation. An unbounded family of such orders gives Sobolev preservation. Specifically, at every tested order the semigroup restricted to the quantum Sobolev space (defined in \eqref{eq:sobolev-norm}) retains its semigroup properties while satisfying
    \begin{equation*}
        \tr\bigl[\rho(t)(N + \1)^{k}\bigr] \leq e^{\omega_k t}\tr\bigl[\rho_0(N + \1)^{k}\bigr]\,,
    \end{equation*}
    for all $t\geq0$. This property provides crucial a priori estimates for perturbation theory on QMSs, as employed in our prior works \cite{Moebus.2024,Moebus.2026Zeno,Moebus.2024Learning,Moebus.2025}. 
    
    Related control of energy growth is studied in the framework of energy-limited quantum dynamics \cite{vanLuijk.2025}. These estimates bound output energy in terms of input energy. However, in many physically relevant examples, a stronger inequality holds:
    \begin{equation}\label{eq:assum-improved-stability}
        \tr\bigl[\cL(\rho)(N+\1)^{k}\bigr] \leq -\mu_{k}\tr\bigl[\rho(N+\1)^{k+\delta}\bigr] + c_{k}\,,
    \end{equation}
    We call \eqref{eq:assum-improved-stability} the \emph{Sobolev regularization condition}. In the examples below it holds for an unbounded family of orders $k>0$, with corresponding constants $\mu_k,\delta>0$ and $c_k\geq0$, and finite-Fock-support states. As conjectured in \cite{Gondolf.2024}, this condition implies not only the boundedness of moments but also their improvement over time. The derivation of this regularization mechanism was developed in parallel with the doctoral thesis \cite[Thm.~3.4.8 and Sec.~3.4.3, Eq.~(3.44)]{Gondolf.2025thesis}. The present work integrates these foundational single-mode results into a broader framework, extending the analysis to multi-mode systems and rigorously establishing the consequences for perturbative and dynamical behavior.
    In Section~4 of \cite{Robin.2025Galerkin}, Robin and Rouchon ask for conditions ensuring that, for an arbitrary initial state $\rho_0\in\cT_1$ and every $t>0$, the evolved state $\rho(t)$ belongs to $W^{k,1}$, singling out two-photon dissipation as a candidate regularizing dynamics. The present results answer this question affirmatively for the class satisfying \eqref{eq:assum-improved-stability}. In particular, Theorem~\ref{thm:sobolev-regularization} and Lemma~\ref{lem:l-diss} prove instantaneous Sobolev regularization for $\cL[a^2]$ and provide an explicit $1\to W^{k,1}$ bound.

    In more detail, we identify and analyze a broad class of unbounded GKSL generators satisfying \eqref{eq:assum-improved-stability} on the bosonic Fock space, whose associated QMSs exhibit a form of regularization, which we call \emph{instantaneous Sobolev regularization}. That is, for any strictly positive time $t > 0$, the evolution maps arbitrary initial states to states with finite moments of all orders, with explicit quantitative control of how these moments behave as~$t \downarrow 0$. This phenomenon explains stability mechanisms in dissipative bosonic dynamics and introduces new analytic tools for quantifying error suppression in bosonic quantum information protocols.

    This instantaneous regularization is of a similar nature as the regularization property of the heat flow (see for example \cite{Kato.1995}): dissipative polynomial Lindbladians immediately generate finite moments of arbitrary order, paralleling the Sobolev smoothing of the classical heat semigroup.

    Such regularization is not automatic, even for the well-studied class of Gaussian quantum Markov semigroups generated by quadratic Hamiltonians and linear Lindblad operators \cite{Agredo.2021,Cipriani.2000}.

    \begin{ex}[Ordinary photon loss]\label{ex:photon-loss-no-regularization}
        Ordinary photon loss preserves finite moments but does not exhibit instantaneous Sobolev regularization. For $T_t=e^{t\cL[a]}$ and $\rho_n=\ketbra{n}{n}$, one has
        \begin{equation*}
            \tr[(N+\1)T_t(\rho_n)]=1+e^{-t}n,
        \end{equation*}
        and hence $\|T_t\|_{1\to W^{1,1}}=\infty$ for every finite $t>0$. Moreover, for a diagonal state $\rho=\sum_{n\geq0}p_n\rho_n$ with $\sum_{n\geq0}np_n=\infty$, positivity and monotone convergence give $\tr[(N+\1)T_t(\rho)]=\infty$. Thus an infinite first moment remains infinite, confirming at this order the absence of regularization anticipated by Robin and Rouchon \cite[Sec.~4]{Robin.2025Galerkin}.
    \end{ex}

    By contrast, the multiphoton dissipators considered below yield moment bounds independent of the initial state at every positive time.

    Technically, our approach builds on the generation and a priori estimate framework of \cite{Gondolf.2024}, together with its extensions in \cite{Moebus.2025thesis}. For the multi-mode extension, we prove a generation theorem for the single-sandwich norm associated with the concentrated observable $M_{v,k}=\sum_i\alpha_{v,i}(N_i+\1)^k$. Concentrated Sobolev stability at a target order and at one sufficiently high auxiliary order implies generation and the weighted domain properties needed for regularization. The theorem is stated in Theorem~\ref{thm:concentrated-generation}, with its proof in Appendix~\ref{appx:concentrated-generation}. Uniform constants in the concentrated Sobolev regularization condition and weight normalization give local moment bounds independent of the finite lattice cardinality. The generation theorem itself does not require locality.

    The paper is organized as follows. In Section~\ref{sec:prelim}, we introduce the mathematical preliminaries. The main technical results are then developed in Section~\ref{sec:sobo-regularizing-qms}, first for single-mode systems in Section~\ref{subsec:single-mode-regularizing} and subsequently extended to multi-mode systems in Section~\ref{subsec:multi-mode-regularizing}, with particular attention to locality. Perturbative applications to the bosonic cat code are discussed in Section~\ref{sec:single-mode-applications}, followed by multi-mode applications in Section~\ref{sec:multi-mode-applications}, and we conclude in Section~\ref{sec:conclusion}. Technical proofs and auxiliary results are collected in Appendices \ref{appx:ode-comparison-lemma}--\ref{appx:concentrated-generation}.

\section{Preliminaries}\label{sec:prelim}
    We begin by recalling the basic setting and assumptions used throughout this work. Let $\mathcal{H}$ denote the bosonic Fock space obtained as the closure of the linear span of the orthonormal basis $\{\ket{n}\}_{n \in \mathbb{N}_0}$ with respect to the standard inner product $\langle \cdot, \cdot \rangle$. The annihilation and creation operators $a$ and $a^\dagger$ act on the Fock basis as $a\ket{n} = \sqrt{n}\ket{n-1}$, $a^\dagger\ket{n} = \sqrt{n+1}\ket{n+1}$, and satisfy the canonical commutation relation $[a, a^\dagger] = \1$. The number operator is given by $N = a^\dagger a$. For systems with $m$ bosonic modes, we use $\mathcal{H}$ for the $m$-fold tensor product of the single-mode space and define $a_i$, $a_i^\dagger$ and $N_i = a_i^\dagger a_i$ on the $i$th factor, for $i \in \{1, \ldots, m\}$. We denote by $\cT_1$ the Banach space of trace-class operators on $\mathcal{H}$ and by $\cT_f = \mathrm{span}\{\ketbra{\boldsymbol n}{\boldsymbol n'} : \boldsymbol n, \boldsymbol n' \in \mathbb{N}_0^m\}$ the subspace of operators with finite support in the joint Fock basis, with $m=1$ in the single-mode case (see \cite{Holevo.2011,Bratteli.1981,Bratteli.1987} for details).
    The finite-Fock vector space is
    \begin{equation*}
        \mathcal H_f=\operatorname{span}\{\ket{\boldsymbol n}:\boldsymbol n\in\mathbb N_0^m\},\qquad \ket{\boldsymbol n}=\ket{n_1}\otimes\cdots\otimes\ket{n_m},
    \end{equation*}
    with $m=1$ in the single-mode case.

    For $\beta \in \mathbb{C}$, the normalized coherent vector is
    \begin{equation*}
        \ket{\beta} = e^{-|\beta|^2/2} \sum_{n=0}^{\infty} \frac{\beta^n}{\sqrt{n!}} \ket{n}.
    \end{equation*}
    For $\alpha = (\alpha_1, \ldots, \alpha_m) \in \mathbb{C}^m$, we write $\ket{\alpha} = \bigotimes_{i=1}^m \ket{\alpha_i}$.

    A \emph{quantum Markov semigroup} (QMS) $(T_t)_{t \ge 0}$ is a family of completely positive, trace-preserving maps, which is additionally $C_0$-semigroup on $\cT_1$, i.e.,
    \begin{equation}\label{eq:semigroup}
        T_t \circ T_s = T_{t+s}, \quad T_0 = \mathrm{id}, \quad \text{and} \quad \lim_{t \downarrow 0} T_t\rho = \rho
    \end{equation}
    for all $\rho \in \cT_1$ and $t,s \ge 0$, with the limit taken in trace norm. Such a semigroup is uniquely determined by its (possibly unbounded) generator $\cL$, defined by
    \begin{equation*}
        \cL(X) = \lim_{t \to 0^+} \frac{1}{t}\big(T_t(X) - X\big),
    \end{equation*}
    for $X$ in its domain $\mathcal{D}(\cL) = \{X \in \cT_1 : t \mapsto T_t(X) \text{ is differentiable on } \mathbb{R}_{\ge 0}\}$. The generator is in both the bounded and unbounded case densely defined and closed. Being a \emph{closed} operator means that the graph $\{(X, \cL(X)) : X \in \mathcal{D}(\cL)\}$ is closed in $\cT_1 \times \cT_1$ (see \cite[Ch.~III]{Kato.1995} and \cite[Thm.~II.1.4]{Engel.2000}).
    If $X\in\cD(\cL)$, then $e^{t\cL}X\in\cD(\cL)$ for every $t\geq0$, the orbit $t\mapsto e^{t\cL}X$ is continuously differentiable, and $\frac{d}{dt}e^{t\cL}X=\cL e^{t\cL}X=e^{t\cL}\cL X$. The same statement holds in every admissible Sobolev space for the generator of the restricted semigroup.

    On finite-dimensional Hilbert spaces, the GKSL theorem \cite{Gorini.1976, Lindblad.1976} guarantees that $\cL$ admits the structure presented in \eqref{eq:GKSL}. This representation motivates the following assumption.
    
    \begin{assum}[Polynomial GKSL generator]\label{assum:concentrated-polynomial}
        The operator $\cL : \cT_f \to \cT_f$ has the form
        \begin{equation}\label{eq:bosonic-GKLS}
            \cL(\rho) \coloneqq -i[H, \rho] + \sum_{i = 1}^{J} \left(L_i \rho L_i^\dagger - \tfrac{1}{2}\{L_i^\dagger L_i, \rho\}\right),
        \end{equation}
        where $J < \infty$, $H,L_1,\ldots,L_J$ are polynomials in
        $\{a_j, a_j^\dagger\}_{j=1}^m$, and $H$ is symmetric on
        $\mathcal H_f$. Adjoints are formal adjoints on $\mathcal H_f$.
        Fix an integer $d \geq 1$ bounding the total degrees of these
        polynomials. No locality assumption is imposed.
    \end{assum}
    
    We write $\cL[J](x) = JxJ^\dagger - \tfrac12\{J^\dagger J,x\}$ for a single jump. In the multi-mode applications, we restrict to local polynomial interactions of the form $(a_i^\dagger)^{k_i} a_i^{\ell_i} (a_j^\dagger)^{k_j} a_j^{\ell_j}$, corresponding to nearest-neighbour couplings on an interaction graph $(V,E)$ with $|V| = m$ and $(i,j)\in E$. Note, however, that an operator $\cL$ of this form need not generate a well-defined QMS on $\cT_1$.

    We use $\dagger$ for Hilbert-space adjoints and $\Phi^*$ for the trace dual, characterized by $\tr[X\Phi(Y)] = \tr[\Phi^*(X)Y]$ on $\cT_f$. For adjoint-preserving maps, this agrees with the Hilbert--Schmidt adjoint: $\tr[X^\dagger\Phi(Y)] = \tr[(\Phi^*(X))^\dagger Y]$ for $X,Y \in \cT_f$.

    A general framework for verifying when operators of the form \eqref{eq:bosonic-GKLS} generate a QMS was developed in \cite{Gondolf.2024}. For $k\geq0$, define $W^{k,1}\subseteq\cT_1$ to be the weighted trace-class space
    \begin{equation*}
        W^{k,1} \coloneqq \left\{ X \in \cT_1 : \|(N + \1)^{k/2} X (N + \1)^{k/2}\|_1 < \infty \right\},
    \end{equation*}
    equipped with the norm
    \begin{equation}\label{eq:sobolev-norm}
        \|X\|_{W^{k,1}} \coloneqq \|(N + \1)^{k/2} X (N + \1)^{k/2}\|_1.
    \end{equation}
    The notation $\|\cdot\|_{W^{k,1} \to W^{k',1}}$ denotes the corresponding operator norm. Related all-orders operator regularity is studied in the theory of Schwartz operators \cite{Keyl.2016}. Here we use weighted trace-class norms to quantify moment production. The space $W^{k,1}$ consists of all $X\in\cT_1$ for which the form $(N+\1)^{k/2}X(N+\1)^{k/2}$, initially defined on finite-Fock vectors, extends to a trace-class operator. Equivalently, $X=(N+\1)^{-k/2}Y(N+\1)^{-k/2}$ for some $Y\in\cT_1$. This identifies $W^{k,1}$ isometrically with the trace class and proves completeness.
    
    \begin{defi}[Admissible subspace]\label{def:admissible-subspace}
        Let $(T_t)_{t\geq0}$ be a $C_0$-semigroup on $\cT_1$, and let
        $(Y,\|\cdot\|_Y)$ be a Banach space continuously embedded in
        $\cT_1$. We call $Y$ \emph{admissible} for $(T_t)_{t\geq0}$ if
        it is invariant,
        \begin{equation*}
            T_t(Y)\subseteq Y,\qquad t\geq0,
        \end{equation*}
        and the restrictions $(T_t|_Y)_{t\geq0}$ is a $C_0$-semigroup
        on $Y$. Equivalently, $T_t|_Y\in\cB(Y)$ for all $t\geq0$ and
        \begin{equation*}
            \lim_{t\downarrow0}\|T_ty-y\|_Y=0,
            \qquad \text{for all}\qquad y\in Y.
        \end{equation*}
    \end{defi}

    \begin{defi}[Sobolev preservation]\label{def:sobolev-preserving-QMS}
        For $k\geq0$, a QMS is called \emph{$k$-Sobolev preserving}
        if $W^{k,1}$ is admissible in the sense of
        Definition~\ref{def:admissible-subspace}. It is called \emph{Sobolev preserving}
        if this holds for $k$ in an unbounded subset of $(0,\infty)$.
    \end{defi}
    By \cite[Prop.~I.5.5]{Engel.2000}, admissibility implies that there
    are constants $C\geq1$ and $\omega\in\mathbb R$ such that
    \begin{equation}\label{eq:single-mode-admissible-semigroup-bound}
        \|T_t\|_{W^{k,1}\to W^{k,1}}\leq Ce^{\omega t},
        \qquad t\geq0.
    \end{equation}
    The following generation theorem proves the sharper estimate with
    $C=1$ and the exponent supplied by the \emph{Sobolev stability condition}. For a closed operator $A$ on a Banach space $X$, a subspace $D_0 \subseteq \mathcal D(A)$ is a \emph{core} if it is dense in
    $\mathcal D(A)$ for the graph norm $\|x\|_X + \|Ax\|_X$.

    \begin{stabilityassum}[Sobolev stability condition]\label{assum:single-mode-stability}
        At a specified order $s \geq 0$, there is a constant
        $\omega_s \geq 0$ such that
        \begin{equation}\label{eq:omega-bound}
            \tr[\cL(\rho)(N + \1)^s]
            \leq \omega_s \tr[\rho(N + \1)^s],
            \qquad \rho \in \cT_f,\quad \rho \geq 0,\quad \tr[\rho] = 1.
        \end{equation}
    \end{stabilityassum}

    \begin{thm}[Single-mode generation from two orders, cf.~{\cite[Thm.~3.1 and Rem.~5]{Gondolf.2024}}]\label{thm:single-mode-generation}
        Let $(\cL,\cT_f)$ satisfy Assumption~\ref{assum:concentrated-polynomial}
        in one mode. Suppose Assumption~\ref{assum:single-mode-stability} holds
        at two real orders $k \geq 0$ and $l > k + 2d$.
        Then the trace-class closure of $(\cL,\cT_f)$ generates a
        QMS $(T_t)_{t\geq0}$. Both $W^{k,1}$ and $W^{l,1}$ are
        admissible in the sense of Definition~\ref{def:admissible-subspace}, with
        \begin{equation}\label{eq:single-mode-generation-bound}
            \|T_t\|_{W^{s,1}\to W^{s,1}}\leq e^{\omega_s t},
            \qquad s\in\{k,l\},\quad t\geq0.
        \end{equation}
        At the target order $k$, the restricted generator is the
        closure of $(\cL,\cT_f)$ in $W^{k,1}$. Thus $\cT_f$ is a
        generator core in trace class and in $W^{k,1}$.
        At $k=0$, the Sobolev stability condition holds automatically with $\omega_0=0$.
    \end{thm}
    \begin{proof}
        The Sobolev order $2k$ in the cited theorem corresponds to
        order $k$ here. The present formulation, including complete
        positivity and the stated generator-core assertions, follows from
        Theorem~\ref{thm:concentrated-generation} for one mode and $M_{v,k}=(N+\1)^k$.
    \end{proof}

    \begin{cor}[Single-mode generation at unbounded orders]\label{cor:single-mode-generation-unbounded}
        Suppose Assumption~\ref{assum:concentrated-polynomial} holds in one mode
        and Assumption~\ref{assum:single-mode-stability} holds at every order in an
        unbounded set $K\subset(0,\infty)$. Then the trace-class closure
        generates a QMS. For every $s\in K$, and every additional order
        $s\geq0$ satisfying \eqref{eq:omega-bound}, $W^{s,1}$ is admissible,
        the bound \eqref{eq:single-mode-generation-bound} holds, and its
        generator is the closure of $(\cL,\cT_f)$ in that space.
    \end{cor}
    \begin{proof}
        For each target order $s$, choose $l\in K$ with $l>s+2d$ and
        apply Theorem~\ref{thm:single-mode-generation}. The resulting semigroups
        agree with the trace-class semigroup.
    \end{proof}

    Many physically relevant unbounded models satisfy the following stronger condition.
    
    \begin{regularizationassum}[Sobolev regularization condition]\label{assum:single-mode-regularization}
        At a specified order $k > 0$, there exist constants $\mu_k > 0$,
        $c_k \geq 0$, and $\delta > 0$ such that
        \begin{equation}\label{eq:strong-bound}
            \tr[\cL(\rho)(N + \1)^k] \leq -\mu_k \tr[\rho (N + \1)^{k + \delta}] + c_k,
            \qquad \rho \in \cT_f,\quad \rho \geq 0,\quad \tr[\rho] = 1.
        \end{equation}
    \end{regularizationassum}

    Once the Sobolev regularization condition has been transferred from $\cT_f$ to trajectories, as in Lemma~\ref{lem:core-drift-transfer}, the argument of \cite[Prop.~4.1]{Gondolf.2024} gives the uniform bound
    \begin{equation}\label{eq:uniform-bound}
        \|e^{t\cL}\|_{W^{k,1} \to W^{k,1}} \le \max\{1, c_k / \mu_k\},
    \end{equation}
    which holds independently of time. In the present work, we extend this result by showing that for generators satisfying the Sobolev regularization condition \eqref{eq:strong-bound}, the associated semigroup exhibits \emph{instantaneous improvement of Sobolev regularity}, that is, the production of higher moments at arbitrarily small positive times. Moreover, we extend the single-mode results to multi-modes and present several applications.

    Algebraic identities are initially understood on the finite-Fock domain $\cT_f$. Whenever an estimate on $\cT_f$ is used along a semigroup trajectory, its extension to the relevant closed-generator domain is justified explicitly. Such an extension is not automatic from density for unbounded moments. The adjoint semigroup on bounded observables is understood in trace duality and is weak* continuous.

\addtocontents{toc}{\protect\setcounter{tocdepth}{2}}
\section{Sobolev regularizing QMS}\label{sec:sobo-regularizing-qms}
    In this section, we present our main technical results establishing \emph{instantaneous Sobolev regularization}. We note that the fundamental derivation for the single-mode case (Section~\ref{subsec:single-mode-regularizing}) was developed in parallel with the results that appeared in the dissertation of one of the authors, \cite[Thm.~3.4.8]{Gondolf.2025thesis}. Here, we present these arguments in a self-contained manner before generalizing them to the multi-mode setting in Section~\ref{subsec:multi-mode-regularizing} and applying them to stability problems.

    We begin with an extension of Gr\"{o}nwall's inequality \cite{Gronwall.1919}, which serves as a key analytical tool for handling Assumption~\ref{assum:single-mode-regularization} through an application of Jensen's inequality.

    \begin{lem}\label{lem:gronwall-extension}
        Let $y : [0, \infty) \to [0, \infty)$ be a locally absolutely continuous function satisfying the differential inequality
        \begin{equation*}
            \frac{dy}{dt}(t) \le -a\, y(t)^p + b
        \end{equation*}
        almost everywhere, for constants $a>0$, $b\geq0$ and exponent $p > 1$. Then,
        \begin{equation*}
            y(t) \le \bigg(\max\{y(0) - \Big(\frac{b}{a}\Big)^{1/p}, 0\}^{1 - p} + a(p - 1)t \bigg)^{\!-\frac{1}{p - 1}} + \Big(\frac{b}{a}\Big)^{\!\frac{1}{p}} =: z(t) \, .
        \end{equation*}
        Here $0^{1-p}=+\infty$ and $(+\infty)^{-1/(p-1)}=0$.
    \end{lem}
    \begin{proof}
        The proof can be found in Lemma~\ref{appx-lem:gronwall-extension}.
    \end{proof}

    \subsection{Single-mode regularizing QMS}\label{subsec:single-mode-regularizing}
        In this section, we apply the tools developed above to establish Sobolev regularization at a fixed order $k$.
        
        \begin{defi}[Sobolev regularization]\label{def:sobolev-regularizing-QMS}
            Fix $k > 0$. A QMS $(T_t)_{t \geq 0}$ is called \emph{instantaneously Sobolev regularizing at order $k$}, or \emph{$k$-Sobolev regularizing}, if $T_t$ is a bounded map from $\cT_1$ to $W^{k,1}$ for every $t > 0$, that is,
            \begin{equation}\label{eq:sobolev-regularizing}
                \|T_t\|_{1 \to W^{k,1}} < \infty,\qquad t > 0.
            \end{equation}
            We call the QMS \emph{Sobolev regularizing}, or say that it exhibits \emph{instantaneous Sobolev regularization at all orders}, if it is $k$-Sobolev regularizing for every $k > 0$. By the embeddings $W^{q,1} \hookrightarrow W^{k,1}$ for $q \geq k$, proved in the stronger form of Corollary~\ref{cor:concentrated-regularization-interpolation}, it suffices to verify this property at an unbounded set of orders.
        \end{defi}
        
        Define
        \begin{equation*}
            y(t) = \| e^{t \cL}(\rho) \|_{W^{k,1}} \, .
        \end{equation*}
        Our goal is to apply Assumption~\ref{assum:single-mode-regularization}. Since this assumption does not exactly fit the form required by Lemma~\ref{lem:gronwall-extension}, we first state and prove the following variant of Jensen's inequality adapted to number-operator moments.

        \begin{lem}[Jensen's inequality for moments]\label{lem:jensens-inequality-moments-number-operator}
            Let $p \ge q > 0$ and let $\rho\in W^{p,1}$ be a state. Then
            \begin{equation}\label{eq:jensen-moment}
                \tr[\rho (N + \1)^p] \ge \big( \tr[\rho (N + \1)^q] \big)^{\frac{p}{q}} \, .
            \end{equation}
        \end{lem}

        \begin{proof}
            Jensen's inequality shows
            \begin{equation*}
                \tr[\rho (N + \1)^p] \coloneqq \sum_{n\geq0} p_n (n + 1)^{q \frac{p}{q}}\geq\Bigl( \sum_{n} p_n (n + 1)^{q} \Bigr)^{\frac{p}{q}}=\big( \tr[\rho (N + \1)^q] \big)^{\frac{p}{q}} \, .
            \end{equation*}
            Here, we used that $\{p_n = \braket{n, \rho n}\}_{n\geq0}$ is a probability distribution and $x \mapsto x^{\frac{p}{q}}$ is convex for $p \ge q > 0$.
        \end{proof}

        Combining Lemma~\ref{lem:gronwall-extension} with the above Jensen-type inequality yields the following result.

        \begin{thm}\label{thm:sobolev-regularization}
            Fix $k,\mu_k,\delta>0$ and $c_k\geq0$, and let $(e^{t\cL})_{t\geq0}$ be a QMS. Suppose that, for every state $\rho\in\cT_f$, the moment $t\mapsto\tr[e^{t\cL}(\rho)(N+\1)^k]$ is finite and locally absolutely continuous on $[0,\infty)$, that $e^{t\cL}(\rho)\in W^{k+\delta,1}$ for almost every $t>0$, and that
            
            \begin{equation}\label{eq:differential-moment-inequality}
                \frac{d}{dt}\tr[e^{t\cL}(\rho)(N+\1)^k]
                \leq-\mu_k\tr[e^{t\cL}(\rho)(N+\1)^{k+\delta}]+c_k
            \end{equation}
            
            almost everywhere. Then, for all $t > 0$,
            \begin{equation}\label{eq:sobolev-regularization-bound}
                \| e^{t \cL} \|_{1 \to W^{k,1}} \le \left( \frac{k}{\delta \mu_k t} \right)^{\frac{k}{\delta}} + \left( \frac{c_k}{\mu_k} \right)^{\frac{k}{k + \delta}} .
            \end{equation}
            In particular, the QMS is $k$-Sobolev regularizing in the sense of Definition~\ref{def:sobolev-regularizing-QMS}.
        \end{thm}
        \begin{proof}
            First let $\rho\in\cT_f$ be a state. Let $\rho(t) = e^{t\cL}(\rho)$. Applying the differential assumption \eqref{eq:differential-moment-inequality} yields almost everywhere
            \begin{equation*}
                \frac{d}{dt} \tr[\rho(t) (N + \1)^k] \le - \mu_k \, \tr[\rho(t) (N + \1)^{k + \delta}] + c_k \, .
            \end{equation*}
            By Lemma~\ref{lem:jensens-inequality-moments-number-operator} with $p = k + \delta$ and $q = k$, we obtain
            \begin{equation*}
                \frac{d}{dt} \tr[\rho(t) (N + \1)^k] \le - \mu_k \, \big( \tr[\rho(t) (N + \1)^k] \big)^{\frac{k + \delta}{k}} + c_k \, .
            \end{equation*}
            Setting $y(t) \coloneqq \tr[\rho(t) (N + \1)^k]$, we have
            \begin{equation*}
                y'(t) \le - \mu_k \, y(t)^{1 + \frac{\delta}{k}} + c_k \, .
            \end{equation*}
            This is of the form required by Lemma~\ref{lem:gronwall-extension} with parameters $a = \mu_k$, $b = c_k$, and $p = 1 + \frac{\delta}{k} > 1$. Applying that lemma gives
            \begin{equation*}
                y(t) \le \big(\max\{ y(0) - (\tfrac{c_k}{\mu_k})^{\frac{k}{k + \delta}}, 0 \}^{-\frac{\delta}{k}} + \tfrac{\delta}{k} \mu_k t \big)^{-\frac{k}{\delta}} + \left( \frac{c_k}{\mu_k} \right)^{\frac{k}{k + \delta}}\,.
            \end{equation*}
            Using the estimate
            \begin{equation*}
                \big( \max\{ y(0) - (\tfrac{c_k}{\mu_k})^{\frac{k}{k + \delta}}, 0 \}^{-\frac{\delta}{k}} + \tfrac{\delta}{k} \mu_k t \big)^{-\frac{k}{\delta}}\le \left( \frac{k}{\delta \mu_k t} \right)^{\frac{k}{\delta}} ,
            \end{equation*}
            we obtain the claimed bound \eqref{eq:sobolev-regularization-bound} for finite-Fock-support states. Normalized positive Fock truncations converge in trace norm to every state. Trace-norm continuity of $e^{t\cL}$ and lower semicontinuity of positive moments therefore give the same bound for every state. Finally, the $2\times2$ argument in \cite[proof of Thm.~4.1.1]{Moebus.2025thesis}, with input weight $\1$ and output weight $(N + \1)^k$, extends it to all trace-class inputs without increasing the constant.
        \end{proof}
        
        \begin{thm}[Single-mode generation and regularization]\label{thm:single-mode-generation-regularization}
            Let $(\cL,\cT_f)$ satisfy Assumption~\ref{assum:concentrated-polynomial} in one mode. Suppose Assumption~\ref{assum:single-mode-regularization} holds at orders $k > 0$ and $l$ with
            \begin{equation*}
                l > k + 2d,\qquad l \geq k + \delta,
            \end{equation*}
            where $\delta$ is the gain at order $k$. Then the trace-class closure of $(\cL,\cT_f)$ generates a $k$-Sobolev preserving QMS satisfying \eqref{eq:sobolev-regularization-bound}. In particular, it is $k$-Sobolev regularizing.
        \end{thm}
        \begin{proof}
            Since $(N + \1)^s \geq \1$ for $s \in \{k,l\}$, Assumption~\ref{assum:single-mode-regularization} implies Assumption~\ref{assum:single-mode-stability} at both orders with $\omega_s = c_s$. Thus Theorem~\ref{thm:single-mode-generation} gives generation, the generator-core property in $W^{k,1}$, and admissibility of $W^{l,1}$. Lemma~\ref{lem:core-drift-transfer}, applied in one mode, supplies the trajectory hypotheses of Theorem~\ref{thm:sobolev-regularization}, which finishes the proof.
        \end{proof}
        If Assumption~\ref{assum:single-mode-regularization} holds at an unbounded family of orders, such an auxiliary order $l$ can be chosen for every target order $k$ in that family.

        \begin{rmk}
            The auxiliary order $l$ has two roles. The condition $l>k+2d$ controls the cutoff-resolvent error in the generation proof, while $l\geq k+\delta$ permits the order-$k$ regularization inequality to be transferred to semigroup trajectories. Such higher-order control is sufficient for the argument, not necessary for generation in general. The $(a^\dagger)^2$ pure-birth example in the introduction shows that a formal GKSL expression and trace conservation on $\cT_f$ alone need not yield a trace-preserving semigroup \cite[Example~3.3]{Davies.1977}.
        \end{rmk}

        \begin{rmk}
            The argument above relies solely on Jensen's inequality and the structure of the Sobolev-type norm $\| \cdot \|_{W^{k,1}}$. Therefore, analogous moment-production bounds hold for any functional that satisfies an inequality of the form used in Lemma~\ref{lem:gronwall-extension}.
        \end{rmk}

        \begin{rmk}
            Theorem~\ref{thm:sobolev-regularization} establishes that if \eqref{eq:differential-moment-inequality} holds for some $\delta>0$, we can study the dynamics of the whole set of quantum states, by extending the previous known case on $W^{k,1}$ to $\cT_1$. By comparing the bounds \eqref{eq:uniform-bound} and \eqref{eq:sobolev-regularization-bound}, we notice that if $c_k/\mu_k\neq1$, there exists $t_0>0$ such that for every $t \geq t_0$
            \begin{equation}
                \left( \frac{k}{\delta \mu_k t}\right)^{\frac{k}{\delta}}+\left( \frac{c_k}{\mu_k} \right)^{\frac{k}{k+\delta}}\leq \max\left\{ 1,\frac{c_k}{\mu_k} \right\}\, .
            \end{equation}
            In other words, this new estimate presents a better behavior for large enough $t$. If $c_k/\mu_k=1$, the regularizing bound instead approaches $1$ as $t\to\infty$ and remains larger than $1$ at every finite time.
        \end{rmk}

    \subsection{Multi-mode: concentrated moment propagation bounds}\label{subsec:multi-mode-regularizing}

        Next, we extend the above result to multi-mode systems defined on a connected $D$-dimensional lattice $(V, E)$ with $|V| = m$. We write $\mathrm{dist}(v,i)$ for the shortest-path distance in this graph. As a first step, we introduce a locally concentrated Sobolev space. To that end, we define the following reference operator centered at a mode $v \in V$ with a decay constant $\kappa > 0$:
        
        \begin{equation}\label{eq:multi-mode-sobolev-operator}
            M_{v,k}=\sum_{i\in V}\alpha_{v,i}(N_i+\1)^k,
            \qquad
            \cW_v^k(x)=M_{v,k}^{1/2}xM_{v,k}^{1/2},
        \end{equation}
        \begin{equation}\label{eq:concentration-weights}
            \alpha_{v,i}=\frac{e^{-\kappa\,\mathrm{dist}(v,i)}}{Z_v},
            \qquad k\geq0,
        \end{equation}
        
        for all $x \in \cT_f$ and normalization constant $Z_v=\sum_{j \in V} e^{-\kappa \, \mathrm{dist}(v,j)}$. For lattice families satisfying the sphere-growth summability condition in Lemma~\ref{appx-lem:normalization-bounds}, we write $Z_*(\kappa)$ for the resulting volume-independent upper bound on $Z_v$. The coefficients are strictly positive and sum to one, independently of $k$. For disconnected graphs, the construction can be applied to each connected component, or one can use an everywhere-finite ambient lattice distance.

        The space $W_v^{k,1}$ consists of all $x\in\cT_1$ for which the form $M_{v,k}^{1/2}xM_{v,k}^{1/2}$, initially defined on finite-Fock vectors, extends to a trace-class operator. It is equipped with the norm
        \begin{equation*}
            \|x\|_{W^{k,1}_v} = \| \cW_v^{k}(x) \|_1.
        \end{equation*}
        
        Equivalently, $x=M_{v,k}^{-1/2}yM_{v,k}^{-1/2}$ for some $y\in\cT_1$. The sandwich map is an isometric isomorphism onto $\cT_1$, with a bounded completely positive inverse. Hence $W_v^{k,1}$ is complete. For fixed finite $V$, define the maximum-norm space by
        \begin{equation*}
            W_{\max}^{k,1}=\bigcap_{v\in V}W_v^{k,1},\qquad
            \|x\|_{W_{\max}^{k,1}}
            =\max_{v\in V}\|x\|_{W_v^{k,1}}.
        \end{equation*}
        The norm is constructed in such a way that for positive operators, the single sandwich gives exactly the concentrated moments
        \begin{equation}\label{eq:positive-concentrated-moment}
            \|\rho\|_{W_v^{k,1}}
            =\tr[M_{v,k}\rho]
            =\sum_i\alpha_{v,i}\tr[(N_i+\1)^k\rho].
        \end{equation}

        One key motivation for the structure of the norm was that control on local moments directly bounds concentrated moments without any dependence on system size. For example, for any $\alpha\in\mathbb{C}^m$ and $k,\theta>0$ (see Lemma~\ref{appx-lem:bound-coherent-state})
        \begin{align*}
            \|\ketbra{\alpha}{\alpha}\|_{W^{k,1}_v}
            &\leq\max_{i\in\{1,...,m\}}\tr[\ketbra{\alpha_i}{\alpha_i}(N_i+\1)^k]\\
            &\leq\left(\frac{k}{e\theta}\right)^k
            \exp\left(\theta+\max_i|\alpha_i|^2(e^\theta-1)\right)\,.
        \end{align*}

        \subsubsection{Generation theory}

        Next, we define concentrated Sobolev preservation on the complete single-sandwich space:
        
        \begin{defi}[Concentrated Sobolev preservation]\label{def:concentrated-sobolev-preserving-QMS}
            For $k\geq0$, a QMS is called \emph{$k$-Sobolev preserving} if $W_{\max}^{k,1}$ is admissible in the sense of Definition~\ref{def:admissible-subspace}. Replacing $W_{\max}^{k,1}$ by $W_v^{k,1}$, we call the QMS \emph{$(k,v)$-Sobolev preserving}. We call the QMS \emph{concentrated Sobolev preserving} if $W_{\max}^{k,1}$ is admissible for $k$ in an unbounded subset of $(0,\infty)$.
        \end{defi}
        By \cite[Prop.~I.5.5]{Engel.2000}, admissibility implies that there are constants $C\geq1$ and $\omega\in\mathbb R$ such that
        \begin{equation}\label{eq:concentrated-admissible-semigroup-bound}
            \|T_t\|_{W_{\max}^{k,1}\to W_{\max}^{k,1}}
            \leq Ce^{\omega t},
            \qquad t\geq0,
        \end{equation}
        and likewise for $W_v^{k,1}$. At a fixed center, the generation theorem below proves the sharper estimate with $C=1$ and exponent $\omega_{v,k}$. If the stability conditions at the target and auxiliary orders hold for every center, it proves \eqref{eq:concentrated-admissible-semigroup-bound} with $C=1$ and exponent $\max_{v\in V}\omega_{v,k}$. If the constants $\omega_{v,k}$ can be chosen uniformly over the centers and finite volumes, this exponent is independent of the volume.

        We use Assumption~\ref{assum:concentrated-polynomial} and the following stability condition. Throughout, $V$ is finite and nonempty, and $\cT_f$ is the linear span of the joint Fock matrix units. The weights in \eqref{eq:multi-mode-sobolev-operator} may be any strictly positive coefficients $\alpha_{v,i}$ with $\sum_i\alpha_{v,i}=1$, independent of $k$.

        \begin{stabilityassum}[Concentrated Sobolev stability condition]\label{assum:concentrated-moments}
            For a fixed center $v\in V$ and order $s\geq0$, there is a constant $\omega_{v,s}\geq0$ such that
            \begin{equation}\label{eq:concentrated-core-generation-drift}
                \tr[M_{v,s}\cL(\rho)]
                \leq \omega_{v,s}\tr[M_{v,s}\rho],
                \qquad \rho \in \cT_f,\quad \rho \geq 0,\quad \tr[\rho] = 1.
            \end{equation}
        \end{stabilityassum}

        This is the fixed-center concentrated analogue of Assumption~\ref{assum:single-mode-stability}. It is a condition on the generator. The theorem below establishes Sobolev preservation at the asserted orders for the generated semigroup.

        \begin{thm}[Generation on concentrated Sobolev spaces from two orders]\label{thm:concentrated-generation} 
            Let $(\cL,\cT_f)$ satisfy Assumption~\ref{assum:concentrated-polynomial}. Fix $v\in V$ and suppose Assumption~\ref{assum:concentrated-moments} holds at this center at two real orders $k\geq0$ and $l>k+2d$. Then the closure of $(\cL,\cT_f)$ in $\cT_1$ generates a completely positive, trace-preserving contraction $C_0$-semigroup $(T_t)_{t\geq0}$. For $s\in\{k,l\}$, the space $W_v^{s,1}$ is admissible, with
            \begin{equation}\label{eq:concentrated-generation-bound}
                \|T_t\|_{W_v^{s,1}\to W_v^{s,1}}
                \leq e^{\omega_{v,s}t},
                \qquad t\geq0.
            \end{equation}
            At the target order $k$, the restricted generator is the closure of $(\cL,\cT_f)$ in $W_v^{k,1}$. Thus $\cT_f$ is a generator core in $\cT_1$ and in $W_v^{k,1}$.

            If the stability condition holds at both orders for every $v\in V$, then the preceding centered conclusions hold for every center. Moreover, the spaces $W_{\max}^{s,1}$, $s\in\{k,l\}$, are admissible, with
            \begin{equation}\label{eq:concentrated-max-generation-bound}
                \|T_t\|_{W_{\max}^{s,1}\to W_{\max}^{s,1}}
                \leq e^{(\max_{v\in V}\omega_{v,s})t},\qquad t\geq0.
            \end{equation}
            At the target order $k$, the restricted generator on $W_{\max}^{k,1}$ is the closure of $(\cL,\cT_f)$, so $\cT_f$ is a generator core in $W_{\max}^{k,1}$. At $k=0$, the stability condition holds automatically at every center with $\omega_{v,0}=0$.
        \end{thm}

        \begin{proof}
            See Appendix~\ref{appx:concentrated-generation}.
        \end{proof}

        Uniform bounds on $\omega_{v,k}$ give semigroup bounds uniform over centers and finite volumes, independently of $\omega_{v,l}$. No generator-core property at $l$ is asserted. Taking $k = 0$, stability at one center and one order $l > 2d$ suffices for trace-class generation.

        \begin{cor}[Generation at unbounded orders]\label{cor:concentrated-generation-unbounded}
            Let $(\cL,\cT_f)$ satisfy Assumption~\ref{assum:concentrated-polynomial}. Fix $v\in V$ and suppose Assumption~\ref{assum:concentrated-moments} holds at this center for every order in an unbounded set $K\subset(0,\infty)$. Then the trace-class closure of $(\cL,\cT_f)$ generates a QMS. For every $k\in K$, the space $W_v^{k,1}$ is admissible, satisfies \eqref{eq:concentrated-generation-bound} at $s=k$, and has generator given by the closure of $(\cL,\cT_f)$. The same assertions hold at any additional order at which the stability condition holds at this center.

            If the hypothesis holds for every center, these conclusions hold simultaneously for all $v\in V$. Moreover, $W_{\max}^{k,1}$ is admissible for every $k\in K$, satisfies \eqref{eq:concentrated-max-generation-bound} at $s=k$, and has generator given by the closure of $(\cL,\cT_f)$. The same maximum-norm conclusions hold at any additional order at which the stability condition holds for every center.
        \end{cor}
        \begin{proof}
            For each target order $k$, choose $l\in K$ with $l>k+2d$ and apply Theorem~\ref{thm:concentrated-generation}. For an additional order $s$, choose $l\in K$ with $l>s+2d$ and apply its fixed-center version. All restrictions agree with the same trace-class semigroup because their trace-class generator is the closure of the same pair $(\cL,\cT_f)$.
        \end{proof}

        For one mode, $M_{v,k}=(N+\1)^k$, so these statements also give the single-mode results. An unbounded family supplies arbitrarily high auxiliary orders for the generator-core conclusions. The fixed-order theorem uses only two stability inequalities.

        \subsubsection{Regularization theory}

        \begin{defi}[Concentrated Sobolev regularization]\label{def:concentrated-sobolev-regularizing-QMS}
            Fix $k > 0$ and $v \in V$. A QMS $(T_t)_{t \geq 0}$ is called \emph{$(k,v)$-Sobolev regularizing} if $T_t$ is a bounded map from $\cT_1$ to $W_v^{k,1}$ for every $t > 0$, that is,
            \begin{equation}\label{eq:concentrated-sobolev-regularizing}
                \|T_t\|_{1 \to W_v^{k,1}} < \infty,\qquad t > 0.
            \end{equation}
            Replacing $W_v^{k,1}$ by $W_{\max}^{k,1}$, we call the QMS \emph{$k$-concentrated Sobolev regularizing}. We call it \emph{concentrated Sobolev regularizing} if it is $k$-concentrated Sobolev regularizing for every $k > 0$.
        \end{defi}
        In one mode, these notions agree with Definition~\ref{def:sobolev-regularizing-QMS}. By the contractive embeddings in Corollary~\ref{cor:concentrated-regularization-interpolation}, regularization at an unbounded set of orders suffices for all orders, both at a fixed center and in the maximum norm.

        The following condition strengthens Sobolev stability to obtain regularization.
        \begin{regularizationassum}[Concentrated Sobolev regularization condition]\label{assum:concentrated-regularization}
            For a fixed center $v \in V$ and order $k > 0$, there exist constants $\delta,\mu_k > 0$ and $c_k \geq 0$ such that
            \begin{equation}\label{eq:strong-concentrated-core-drift}
                \tr[M_{v,k}\cL(\rho)]
                \leq -\mu_k \tr[M_{v,k+\delta}\rho]
                + c_k,
                \qquad \rho \in \cT_f,\quad \rho \geq 0,\quad \tr[\rho] = 1.
            \end{equation}
        \end{regularizationassum}
        The regularization theorem below assumes the corresponding trajectory inequality, justified under the generator-domain hypotheses of Lemma~\ref{lem:core-drift-transfer}. We first establish a Jensen-type inequality for multi-mode moments.

        \begin{lem}[Jensen's inequality for multi-mode moments]\label{lem:jensens-inequality-moments-number-operator-multi}
            Fix $v\in V$. Let $p\geq q>0$ and let $\rho\in W_v^{p,1}$ be a state. Then
            \begin{equation}\label{eq:multi-jensen-moment}
                \tr[\cW_v^{p} (\rho) ]\geq \big( \tr[\cW_v^{q} (\rho) ] \big)^{\frac{p}{q}}\,.
            \end{equation}
        \end{lem}

        \begin{proof}
            The proof follows the lines of the proof of Lemma~\ref{lem:jensens-inequality-moments-number-operator}, with slightly more general coefficients $p_n$. Details can be found in Lemma~\ref{appx-lem:jensens-inequality-moments-number-operator-multi} in Appendix~\ref{appx:multi-mode-extension}.
        \end{proof}

        As in Section~\ref{subsec:single-mode-regularizing}, we are now able to combine Lemma~\ref{lem:gronwall-extension} with the above Jensen-type inequality and then afterwards with the above generation theory result:

        \begin{thm}\label{thm:sobolev-regularization-multi-mode}
            Fix $v \in V$, $k,\mu_k,\delta > 0$, and $c_k \geq 0$, and let $(T_t)_{t \geq 0}$ be a QMS. Suppose that, for every state $\rho \in \cT_f$, the moment $t \mapsto \tr[M_{v,k}T_t(\rho)]$ is finite and locally absolutely continuous on $[0,\infty)$, that $T_t(\rho) \in W_v^{k+\delta,1}$ for almost every $t > 0$, and that
            \begin{equation}\label{eq:assum-improved-stability-multi-mode}
                \frac{d}{dt} \tr[M_{v,k}T_t(\rho)]
                \leq -\mu_k \tr[M_{v,k+\delta}T_t(\rho)] + c_k
            \end{equation}
            almost everywhere. Then, for all $t > 0$,
            \begin{equation}\label{eq:quantitative-concentrated-regularization}
                \|T_t\|_{1 \to W_v^{k,1}}
                \leq \left( \frac{k}{\delta \mu_k t} \right)^{\frac{k}{\delta}}
                + \left( \frac{c_k}{\mu_k} \right)^{\frac{k}{k+\delta}}
                \eqqcolon B_k(t).
            \end{equation}
            In particular, the QMS is $(k,v)$-Sobolev regularizing in the sense of Definition~\ref{def:concentrated-sobolev-regularizing-QMS}.

            If the hypotheses hold for every center with the same constants $\delta,\mu_k,c_k$, then
            \begin{equation}\label{eq:multi-mode-max-regularization}
                \|T_t\|_{1 \to W_{\max}^{k,1}} \leq B_k(t),\qquad t > 0.
            \end{equation}
            In particular, the QMS is $k$-concentrated Sobolev regularizing. If these constants are also independent of the finite volume, so is the bound.
        \end{thm}
        \begin{proof}
            The fixed-center bound is Proposition~\ref{appx-thm:sobolev-regularization-multi-mode}, proved by the multi-mode Jensen inequality, Lemma~\ref{lem:gronwall-extension}, and the completely positive extension from states to all trace-class operators. If the hypotheses hold at every center with the same constants, taking the maximum of the centered estimates for each $x \in \cT_1$ proves \eqref{eq:multi-mode-max-regularization}.
        \end{proof}

        The theorem concerns a fixed order $k$ and assumes that the QMS is already given. The constants $\delta,\mu_k,c_k$ specify its quantitative estimate, rather than the definition of regularization. The following theorem combines this result with Theorem~\ref{thm:concentrated-generation}.

        \begin{thm}[Concentrated generation and regularization]\label{thm:concentrated-generation-regularization}
            Let $(\cL,\cT_f)$ satisfy Assumption~\ref{assum:concentrated-polynomial}. Fix $v \in V$ and suppose Assumption~\ref{assum:concentrated-regularization} holds at this center at orders $k > 0$ and $l$ satisfying
            \begin{equation}\label{eq:regularization-auxiliary-order}
                l > k + 2d,\qquad l \geq k + \delta,
            \end{equation}
            where $\delta$ is defined in Assumption~\ref{assum:concentrated-regularization} as the gain at order $k$. Then the trace-class closure of $(\cL,\cT_f)$ generates a $(k,v)$-Sobolev preserving QMS satisfying \eqref{eq:quantitative-concentrated-regularization}. In particular, it is $(k,v)$-Sobolev regularizing.

            If the condition holds at both orders for every center, with the same $\delta,\mu_k,c_k$ at order $k$, then the QMS is $k$-Sobolev preserving and satisfies \eqref{eq:multi-mode-max-regularization}. In particular, it is $k$-concentrated Sobolev regularizing. If these target-order constants are independent of the finite volume, so is the regularization bound.
        \end{thm}
        \begin{proof}
            Since $M_{v,s} \geq \1$ for $s \in \{k,l\}$, Assumption~\ref{assum:concentrated-regularization} implies Assumption~\ref{assum:concentrated-moments} at both orders with $\omega_{v,s} = c_s$. Thus Theorem~\ref{thm:concentrated-generation} gives generation, the generator-core property in $W_v^{k,1}$, and admissibility of $W_v^{l,1}$. Lemma~\ref{lem:core-drift-transfer} supplies the trajectory hypotheses of Theorem~\ref{thm:sobolev-regularization-multi-mode}, which proves the fixed-center conclusion. If the condition holds at every center, the same argument and the maximum-norm conclusions of these theorems give the remaining assertions.
        \end{proof}
        At the auxiliary order $l$, Assumption~\ref{assum:concentrated-moments} alone suffices. If Assumption~\ref{assum:concentrated-regularization} holds at a fixed center for an unbounded family of orders, a suitable $l$ can be chosen for every target order in that family. If it holds at every center for an unbounded family of orders, the generated QMS is concentrated Sobolev regularizing by Corollary~\ref{cor:concentrated-regularization-interpolation}.

        \begin{rmk}
            The fixed-center result needs conditions only at the chosen center, while the maximum-norm result controls all centers simultaneously without a factor of $|V|$.
            For the spatial weights \eqref{eq:concentration-weights}, the fixed-center bound implies, for every state $\rho$,
            \begin{equation}\label{eq:uniform-local-moment}
                \tr[(N_v + \1)^k T_t(\rho)] \leq Z_v B_k(t),\qquad t > 0.
            \end{equation}
            If $Z_v \leq Z_*(\kappa)$ and the regularization constants are uniform over finite volumes, this local-moment bound is volume independent.
        \end{rmk}

        \begin{cor}[Interpolation of concentrated regularization]
            \label{cor:concentrated-regularization-interpolation} Let $0<k<q$, put $\theta=k/q$, and fix $v\in V$. For every $x\in W_v^{q,1}$,
            \begin{equation}\label{eq:fixed-center-concentrated-norm-interpolation}
                \|x\|_{W_v^{k,1}}
                \leq\|x\|_1^{1-\theta}
                \|x\|_{W_v^{q,1}}^\theta.
            \end{equation}
            Consequently, if $(T_t)_{t\geq0}$ is $(q,v)$-Sobolev regularizing, then
            \begin{equation}\label{eq:fixed-center-concentrated-regularization-interpolation}
                \|T_t\|_{1\to W_v^{k,1}}
                \leq\|T_t\|_{1\to W_v^{q,1}}^{k/q},
                \qquad t>0.
            \end{equation}
            Taking the maximum over the centers gives, for every $x\in W_{\max}^{q,1}$,
            \begin{equation}\label{eq:concentrated-norm-interpolation}
                \|x\|_{W_{\max}^{k,1}}
                \leq\|x\|_1^{1-\theta}
                \|x\|_{W_{\max}^{q,1}}^\theta.
            \end{equation}
            Hence, if $(T_t)_{t\geq0}$ is $q$-concentrated Sobolev regularizing, then
            \begin{equation}\label{eq:concentrated-regularization-interpolation}
                \|T_t\|_{1\to W_{\max}^{k,1}}
                \leq\|T_t\|_{1\to W_{\max}^{q,1}}^{k/q},
                \qquad t>0.
            \end{equation}
            The fixed-center statements include the single-mode case. In particular, a volume-independent bound $B_q(t)$ gives the bound $B_q(t)^{k/q}$ at order $k$, without an additional prefactor.
        \end{cor}
        \begin{proof}
            Fix $v$ and write $M=M_{v,q}$. Concavity of $u\mapsto u^\theta$ on the joint Fock basis, together with $\sum_i\alpha_{v,i}=1$, gives $M_{v,k}\leq M^\theta$. The commuting operator $M_{v,k}^{1/2}M^{-\theta/2}$ is therefore a contraction, and
            \begin{equation*}
                \|x\|_{W_v^{k,1}}
                \leq\|M^{\theta/2}xM^{\theta/2}\|_1
                \leq\|x\|_1^{1-\theta}\|x\|_{W_v^{q,1}}^\theta.
            \end{equation*}
            For $x\in\cT_f$, the second inequality follows from the three-lines theorem applied to $z\mapsto M^{z/2}xM^{z/2}$ (cf.~\cite[Thm.~2.14]{Gondolf.2024}). Finite-Fock cutoffs extend it to every $x\in W_v^{q,1}$. Applying it to $T_tx$ and using trace-norm contractivity proves the fixed-center operator bound. Taking the maximum in the pointwise estimate gives the two maximum-norm bounds.
        \end{proof}
        Thus estimates at an unbounded set of orders yield regularization at all lower orders, both at a fixed center and in the maximum norm.

    \subsection{Intermediate and large time perturbation bounds}\label{subsec:perturbation-theory}
        Under a slightly stronger assumption on the degree of dissipation, we can even extend known perturbation-theory bounds for both intermediate and large times to the $1 \rightarrow 1$ norm. To highlight the underlying difficulty, we begin by improving the standard integral equation using our Sobolev regularization property. It is well known that
        \begin{equation*}
            e^{t\cL}(\rho) - e^{t(\cL+\cE)}(\rho) = -t\int_0^1 e^{(1-s)t(\cL+\cE)}\cE\, e^{s t\cL}(\rho)\, ds
        \end{equation*}
        holds for generators $(\cL, \cD(\cL))$ and $(\cL+\cE, \cD(\cL+\cE))$ on $Y$ under the \hyperref[assum:perturbation-common-domain]{common domain assumption} with $\varepsilon = 1$. Even under the relative boundedness assumptions
        \begin{equation*}
            \|\cE(\rho)\|\leq c_1\|\cL(\rho)\|+c_2\|\rho\|
        \end{equation*}
        for $\rho\in\cD(\cL)\subset\cD(\cE)$, one does not expect a substantially better bound than
        \begin{equation*}
            \|e^{t\cL}(\rho) - e^{t(\cL+\cE)}(\rho)\| \leq t \bigl( c_1\|\cL(\rho)\| + c_2\|\rho\| \bigr)\,,
        \end{equation*}
        which depends on the output of an unbounded operator (see, e.g., \cite[Chap.~9]{Kato.1995}). In our setting of Sobolev regularizing quantum Markov semigroups (QMS), we improve the above bound to the uniform topology. We state the result directly in the multi-mode setup, but emphasize that it reduces immediately to the single-mode case used in Section~\ref{subsec:applicatio-perturbation-theory}.

        The improved bounds use the quantitative estimate \eqref{eq:quantitative-concentrated-regularization} with $\delta>k$, so that $B_k$ is integrable near zero, and a perturbation relatively bounded with respect to $\cW^{k}_v$. Throughout this subsection, fix $\varepsilon\geq0$ and let $\cB_\varepsilon$ denote the perturbed generator. All relative Sobolev bounds are assumed for every $x\in W_v^{k,1}\subset\cD(\cE)$. For Duhamel's formula, we impose the following assumption.

        \begin{assum*}[Common domain]\phantomsection\label{assum:perturbation-common-domain}
            There is a Banach space $Y$ continuously and densely embedded in $W_v^{k,1}$ on which both semigroups restrict to $C_0$-semigroups. Moreover, $Y \subset \cD(\cL) \cap \cD(\cB_\varepsilon) \cap \cD(\cE)$, the maps $\cL,\cB_\varepsilon,\cE:Y \to \cT_1$ are bounded, and
            \begin{equation*}
                \cB_\varepsilon x = \cL x + \varepsilon\cE x,
                \qquad x \in Y.
            \end{equation*}
        \end{assum*}
        \begin{rmk}
            For the polynomial perturbations below, the Sobolev regularization condition at a common unbounded set of orders for $\cL$ and $\cB_\varepsilon$ implies high-order admissibility by Corollary~\ref{cor:concentrated-generation-unbounded}. At a sufficiently high common order, $Y=W_v^{r,1}$ satisfies the generator-domain and boundedness properties by Lemma~\ref{lem:core-drift-transfer} and polynomiality. The identity $\cB_\varepsilon=\cL+\varepsilon\cE$ extends from $\cT_f$ to $Y$ by density, so the common-domain assumption holds.
        \end{rmk}
        Both orientations of Duhamel's formula are proved in Lemma~\ref{lem:perturbation-duhamel}.
        
        \begin{prop}[Intermediate-time perturbation bound]\label{prop:intermediate-time-perturbation-bound}
            Let $(\cL,\cD(\cL))$ be the generator of a QMS satisfying \eqref{eq:quantitative-concentrated-regularization} with $T_t=e^{t\cL}$ and $\delta>k$, and let $(\cE,\cD(\cE))$ be an unbounded operator on $\cT_1$ such that $(\cB_\varepsilon,\cD(\cB_\varepsilon))$ generates a contractive semigroup. Suppose that the \hyperref[assum:perturbation-common-domain]{common domain assumption} holds and that
            \begin{equation*}
                \|\cE(\rho)\|_1\leq c_1\|\cW^{k}_v(\rho)\|_1+c_2\|\rho\|_1 = c_1\|\rho\|_{W^{k,1}_v}+c_2\|\rho\|_1\,.
            \end{equation*}

            Then,
            \begin{equation*}
                \|e^{t\cL}-e^{t(\cL+\varepsilon\cE)}\|_{1\rightarrow1}\leq \varepsilon \Bigl(c_1 \frac{\delta}{\delta-k} t^{1-\frac{k}{\delta}}\left(\frac{k}{\delta \mu_k} \right)^{\frac{k}{\delta}} + c_1 t\left( \frac{c_k}{\mu_k} \right)^{\frac{k}{k + \delta}} + c_2 t \Bigr)\,.
            \end{equation*}

        \end{prop}

        \begin{proof}
            First, by the fundamental theorem of calculus, contractivity, and relative boundedness,
            \begin{equation*}
                \begin{aligned}
                    \|e^{t\cL}(\rho)-e^{t(\cL+\varepsilon\cE)}(\rho)\|_1
                    &=\Bigl\|\varepsilon t \int_{0}^1 e^{(1-s)t(\cL+\varepsilon\cE)}\cE\, e^{st\cL}(\rho)\, ds \Bigr\|_1 \\
                    &\leq \varepsilon t \int_{0}^1 \|\cE\, e^{st\cL}(\rho)\|_1\, ds\\
                    &\leq \varepsilon t \int_{0}^1 c_1\| e^{st\cL}(\rho)\|_{W^{k,1}_v}\, ds + \varepsilon t\, c_2\|\rho\|_1
                \end{aligned}
            \end{equation*}
            for $\rho\in Y$. The \hyperref[assum:perturbation-common-domain]{common domain assumption} justifies the fundamental theorem of calculus. Lemma~\ref{lem:perturbation-duhamel} extends the identity to all trace-class inputs. Applying the quantitative estimate \eqref{eq:quantitative-concentrated-regularization} yields
            \begin{equation*}
                \begin{aligned}
                    \|e^{t\cL}(\rho)-e^{t(\cL+\varepsilon\cE)}(\rho)\|_1
                    &\leq \varepsilon t \Bigl(\int_{0}^1 c_1 \left(\frac{k}{\delta \mu_k s t}\right)^{\frac{k}{\delta}} ds + c_1\left( \frac{c_k}{\mu_k} \right)^{\frac{k}{k + \delta}} + c_2\Bigr)\|\rho\|_1 \\
                    &=\varepsilon \Bigl( c_1 \frac{\delta}{\delta-k} t^{1-\frac{k}{\delta}} \left(\frac{k}{\delta \mu_k}\right)^{\frac{k}{\delta}} + c_1 t\left(\frac{c_k}{\mu_k}\right)^{\frac{k}{k+\delta}} + c_2 t\Bigr)\|\rho\|_1\,.
                \end{aligned}
            \end{equation*}
            Since $Y$ is dense in $\cT_1$ and the norms are continuous, closure completes the proof.
        \end{proof}

        One of our main applications, presented in Sections~\ref{sec:single-mode-applications}--\ref{sec:multi-mode-applications}, is the bosonic quantum error-correction cat code. A key ingredient is an underlying dynamics that converges exponentially fast to an invariant subspace --- the codespace. In addition to the perturbation analysis above, a bound on the propagation of small errors for large times is crucial when the dynamics is not implemented perfectly. To this end, we extend the result of Theorem 6 in \cite{Szehr.2013} to the setting of unbounded generators and then verify the corresponding assumptions used in Section~\ref{sec:single-mode-applications} in the case of the bosonic cat code.

        A key ingredient is a strengthened exponential convergence assumption enabled by the Sobolev regularization property. Assuming a $(k,v)$-Sobolev regularizing semigroup converges exponentially fast, that is, there exist a contractive projection $\cP$ and constants $C,\gamma>0$ such that
        
        \begin{equation*}
            \|e^{t\cL}-\cP\|_{W_v^{k,1}\to1}\leq Ce^{-\gamma t},\qquad
            \cP e^{t\cL}=e^{t\cL}\cP=\cP,
        \end{equation*}
        
        we can lift this directly to a $1\rightarrow 1$ trace norm bound:
        
        \begin{equation}\label{eq:regularization-upgrades-mixing}
            \|e^{t\cL}-\cP\|_{1\to1}\leq\widetilde C e^{-\gamma t},\qquad
            \widetilde C=e^{\gamma\tau}\max\{2,C\|e^{\tau\cL}\|_{1\to W_v^{k,1}}\}
            \quad(\tau>0).
        \end{equation}
        For $t\geq\tau$, apply the weighted estimate to $e^{\tau\cL}$. For $t<\tau$, use the bound by $2$. Only boundedness at the fixed time $\tau>0$ is needed. A more detailed discussion in the same spirit appears in Section~\ref{subsec:exp-convergence}. We begin with the case of a unique fixed point/ steady state.

        \begin{prop}[Perturbation bound for steady states]\label{prop:steady-state-perturbation-bound}
            
            Let $(\cL,\cD(\cL))$ and $(\cB_\varepsilon,\cD(\cB_\varepsilon))$ be generators of QMSs satisfying the \hyperref[assum:perturbation-common-domain]{common domain assumption}. Assume that the perturbed semigroup satisfies the quantitative bound \eqref{eq:quantitative-concentrated-regularization} with $\delta>k$, $\mu_k>0$, and $c_k\geq0$, and that
            \begin{equation}\label{eq:perturbation-relative-sobolev-bound}
                \|\cE(\rho)\|_1\leq c_1\|\rho\|_{W_v^{k,1}}+c_2\|\rho\|_1
                \quad(\rho\in W_v^{k,1}).
            \end{equation}
            
            If $e^{t\cL}$ converges to a steady state $\overline{\rho}\in\cT_1$, i.e.,
            
            \begin{equation}\label{eq:uniform-asymptotic-channel-mixing}
                \|e^{t\cL}-\cP_{\overline\rho}\|_{1\to1}\leq\widetilde C e^{-\gamma t},
                \quad t\geq1,\qquad \cP_{\overline\rho}(x)=\tr[x]\overline\rho,
            \end{equation}
            and $\tr[\cE x]=0$ for $x\in W_v^{k,1}$, then the following estimate holds. The trace condition follows from the \hyperref[assum:perturbation-common-domain]{common domain assumption} and \eqref{eq:perturbation-relative-sobolev-bound} when $\varepsilon>0$. For $\varepsilon=0$ it is unnecessary. For any states $\rho,\sigma\in\cT_1$,
            
            \begin{equation*}
                \begin{aligned}
                    \|e^{t\cL}(\sigma)-e^{t(\cL+\varepsilon\cE)}(\rho)\|_1
                    \leq
                    \begin{cases}
                        \|\sigma-\rho\|_1+\varepsilon t^{1-k/\delta}\hat C_1\|\rho\|_1,&t\leq1,\\[4pt]
                        \widetilde C e^{-\gamma t}\|\sigma-\rho\|_1+
                        \varepsilon\hat C_{\rm long}\|\rho\|_1,&t\geq1.
                    \end{cases}
                \end{aligned}
            \end{equation*}

            The constants are
            \begin{equation}\label{eq:perturbation-convolution-constants}
                \begin{gathered}
                    a=k/\delta,\qquad A=c_1\left(\frac{k}{\delta\mu_k}\right)^a,
                    \qquad B=c_1\left(\frac{c_k}{\mu_k}\right)^{k/(k+\delta)}+c_2,\\
                    \hat C_1=\frac{A}{1-a}+B,\qquad
                    \hat C_{\rm long}=\hat C_1+\widetilde C e^{-\gamma}
                    \left[A\left(\gamma^{-1}+(1-a)^{-1}\right)+B\gamma^{-1}\right].
                \end{gathered}
            \end{equation}
            The parameters in the quantitative estimate are those of the perturbed semigroup.
            
        \end{prop}
        \begin{proof}
            We follow the proof of Theorem 6 in \cite{Szehr.2013}, combined with the newly developed techniques for $(k,v)$-Sobolev regularizing QMS. Details are presented in Proposition~\ref{appx-prop:steady-state-perturbation-bound}.
        \end{proof}

        \begin{rmk}
            The time $t^*=1$, which separates the two regimes above, is chosen for convenience. If sharper estimates are needed, $t^*$ can be optimized. The same applies to Proposition~\ref{prop:invariant-subset-perturbation-bound}.
        \end{rmk}

        Continuous quantum error correction relies on the idea that a quantum dynamical process continuously drives the state always back into an invariant subspace --- the codespace. Since the codespace has usually dimension greater than one, we next consider QMS converging to an invariant subspace. Note that perturbations inside the code space may lead to non-convergent rotations, which is why our error metric considers the component complementary to the asymptotic channel. This is not an orthogonal compression onto the complement of the Hilbert-space codespace.

        \begin{prop}[Perturbation bound for invariant subsets]\label{prop:invariant-subset-perturbation-bound}
            
            Let $(\cL,\cD(\cL))$ and $(\cB_\varepsilon,\cD(\cB_\varepsilon))$ be generators of QMSs satisfying the \hyperref[assum:perturbation-common-domain]{common domain assumption}. Assume that the perturbed semigroup satisfies the quantitative bound \eqref{eq:quantitative-concentrated-regularization} with $\delta>k$, $\mu_k>0$, and $c_k\geq0$, and that
            \begin{equation}\label{eq:invariant-relative-sobolev-bound}
                \|\cE(\rho)\|_1\leq c_1\|\rho\|_{W_v^{k,1}}+c_2\|\rho\|_1
                \quad(\rho\in W_v^{k,1}).
            \end{equation}
            
            If $e^{t\cL}$ converges to an invariant subset with contractive projection $\cP$ (with $\cP^\perp=\mathrm{id}-\cP$), i.e.,
            
            \begin{equation*}
                \|e^{t\cL}-\cP\|_{1\to1}\leq\widetilde C e^{-\gamma t},\quad t\geq1,
                \qquad \cP e^{t\cL}=e^{t\cL}\cP=\cP,
            \end{equation*}
            
            Then for any states $\rho,\sigma\in\cT_1$,
            
            \begin{equation*}
                \begin{aligned}
                    \|\cP^\perp e^{t\cL}(\sigma)-\cP^\perp e^{t(\cL+\varepsilon\cE)}(\rho)\|_1
                    \leq
                    \begin{cases}
                        \|\cP^\perp(\sigma-\rho)\|_1+2\varepsilon t^{1-k/\delta}\hat C_1\|\rho\|_1,&t\leq1,\\[4pt]
                        \widetilde C e^{-\gamma t}\|\cP^\perp(\sigma-\rho)\|_1+
                        \varepsilon\hat C_{\rm long}^\perp\|\rho\|_1,&t\geq1.
                    \end{cases}
                \end{aligned}
            \end{equation*}

            Here $a,A,B,\hat C_1$ are as in Proposition~\ref{prop:steady-state-perturbation-bound}, and
            \begin{equation*}
                \hat C_{\rm long}^{\perp}=2\hat C_1+\widetilde C e^{-\gamma}
                \left[A\left(\gamma^{-1}+(1-a)^{-1}\right)+B\gamma^{-1}\right].
            \end{equation*}
            The factor $2$ uses $\|\cP^\perp\|_{1\to1}\leq2$. This complementary projection need not be contractive.
            
        \end{prop}
        \begin{proof}
            As in Theorem 5.2 of \cite{Gondolf.2024}, combined with the $(k,v)$-Sobolev regularizing techniques. Full details appear in Proposition~\ref{appx-prop:invariant-subset-perturbation-bound}.
        \end{proof}

\addtocontents{toc}{\protect\setcounter{tocdepth}{1}}

\section{Single-mode applications}\label{sec:single-mode-applications}
    Supplementary to the above theoretical section, we now present several examples that satisfy the stated assumption. In this section, we focus on the single-mode results combined with exponentially fast convergence to the invariant subset, as well as large-time perturbation bounds in Section~\ref{subsec:single-mode-application}. We would like to emphasize that the obtained results are in the uniform topology and are not fixed to a specific energy level of the system, which is one of the crucial findings. In Section~\ref{sec:multi-mode-applications}, we then extend the applications to multi-mode systems.

    \noindent Note that many results are based on the work \cite{Gondolf.2024} and also part of \cite{Gondolf.2025thesis}.

    \subsection{Families of photon dissipation}\label{subsec:single-mode-application}
        A natural and simple example in quantum optics and superconductors is the $2$-photon dissipation $\cL[a^2]$. For a self-adjoint operator $X$, let $X_+$ denote its positive part. For every state $\rho\in\cT_f$ and $k\geq1$, it satisfies
        \begin{equation*}
            \begin{aligned}
                \tr[\cL[a^2](\rho)(N+\1)^k]&=-\tr[\rho N(N-\1)\bigl((N+\1)^k-(N-\1)_+^k\bigr)]\\
                &\leq-2\tr[\rho(N+\1)^{k+1}]+6\tr[\rho(N+\1)^{k}]\\
                &\leq-\tr[\rho(N+\1)^{k+1}]+6^{k+1}
            \end{aligned}
        \end{equation*}
        The estimate follows from the moment calculation in \cite[Lem.~4.3]{Gondolf.2024}. Generation and instantaneous Sobolev regularization at all orders follow from Lemma~\ref{lem:l-diss} below with $\ell=2$ and $\alpha=0$. Motivated by this simple but important example, we shortly state the bounds for the shifted $\ell$-photon dissipation $\cL_\ell\coloneqq\cL[a^\ell-\alpha^\ell]$, with integer $\ell \geq 2$, key for the bosonic error correction cat code. Due to the shift, it admits the following invariant subspace --- the code space (for $\alpha\neq0$):
        \begin{equation}\label{eq:codespace-ell-photon-diss}
            \mathcal{C}_\ell(\alpha) \coloneqq \mathrm{span}\left\{\ket{\alpha_1}\bra{\alpha_2} : \alpha_1, \alpha_2 \in \left\{\alpha e^{\frac{i2\pi j}{\ell}} : j \in \{0, \ldots, \ell-1\}\right\}\right\}.
        \end{equation}

        For clarity, the corresponding Hilbert-space subspace is the kernel $\mathsf K_\ell(\alpha)=\ker(a^\ell-\alpha^\ell)$, and $\mathcal C_\ell(\alpha)=P_\ell\cT_1P_\ell$, where $P_\ell$ is the orthogonal projection onto $\mathsf K_\ell(\alpha)$. For $\alpha=0$, use $\mathsf K_\ell(0)=\operatorname{span}\{\ket0,\ldots,\ket{\ell-1}\}$ instead of the coherent-state formula. The code space is protected due to the exponential convergence analyzed in Section~\ref{subsec:exp-convergence}. A detailed analysis of quantum error correction (QEC) and construction of the above dynamics can be found in \cite[Secs.~2.1 and 3]{Guillaud.2023} and \cite{Azouit.2016, Gondolf.2024}.

        \begin{lem}[Pure photon dissipation]\label{lem:l-diss}
            For any real $k \geq 1$, integer $\ell \geq 2$, $\alpha \in \mathbb{C}$ and any state $\rho \in \cT_f$,
            \begin{align*}
                \tr\big[\cL_\ell(\rho)(N+\1)^{k}\big]&\le -\frac{\ell}{2} \tr\big[\rho\,(N+\1)^{k+\ell-1}\big]+\frac{\ell}{2}c_k^{(\ell)}\,,
            \end{align*}
            where $c_k^{(\ell)}=\Delta_\ell^\nu\left(\frac{(\nu-1)^{\nu-1}}{\nu^\nu}\right)$ with $\nu=\ell+k-1$ and $\Delta_\ell=(\ell+1)\ell+4|\alpha|^\ell k\ell^{k - 1}\sqrt{\ell!}$. Therefore, $\cL_\ell$ generates a Sobolev regularizing QMS satisfying
            \begin{equation*}
                \| e^{t \cL_\ell} \|_{1 \to W^{k,1}} \le
                \left(\frac{2k}{\ell(\ell-1)t}\right)^{k/(\ell-1)}
                +\bigl(c_k^{(\ell)}\bigr)^{k/(k+\ell-1)} .
            \end{equation*}

        \end{lem}
        \begin{proof}
            The proof is a direct application of Theorem~\ref{thm:sobolev-regularization} to Lemma 4.3 in \cite{Gondolf.2024}. The source uses moment order $k/2$, so substitute $2k$ there. The same Sobolev regularization condition on $\cT_f$ at every integer order implies the Sobolev stability condition and first gives generation by Corollary~\ref{cor:single-mode-generation-unbounded}, and Lemma~\ref{lem:core-drift-transfer} justifies the moment inequality along the generated trajectories. The regularization estimate holds for every real $k\geq1$. Corollary~\ref{cor:concentrated-regularization-interpolation} also gives quantitative regularization bounds for $0<k<1$.
        \end{proof}
        In the context of the bosonic error correction cat code, these models normally come with a Hamiltonian under certain degree assumption. The Hamiltonian is of the following structure
        \begin{equation}\label{eq:single-mode-H}
            H=\sum_{\substack{0\leq i \le j\\i+j \le d_H}}\left(\lambda_{i,j}a^i(a^\dagger)^j + \overline{\lambda_{i,j}} a^j (a^\dagger)^i\right) \, .
        \end{equation}
        for coefficients $\lambda_{i,j}\in\mathbb C$. With the assumption that $d_H \leq 2(\ell-1)$, we achieve the following Sobolev regularization result.

        \begin{lem}\label{lem:l-diss-hamiltonian}
            Let $k \geq 1$ be real, $\ell \geq 2$ be an integer, $\alpha \in \mathbb{C}$, and let $H$ be the Hamiltonian \eqref{eq:single-mode-H} with $d_H\leq 2(\ell-1)$. Then, for all states $\rho\in\cT_f$
            \begin{equation}
                \begin{aligned}
                    \tr[(-i[H,\cdot] + \cL_\ell)(\rho)(N+\1)^{k}] &\leq-\frac{\ell}{2}\,\tr[\rho(N+\1)^{\ell+k-1}]+\frac{\ell}{2}c_{k,H}^{(\ell)}\,.\label{eqdiffLH}
                \end{aligned}
            \end{equation}
            where $c_{k,H}^{(\ell)}$ and $\nu$ are defined by
            \begin{equation*}
                \begin{gathered}
                    c_{k,H}^{(\ell)}=c^\nu\frac{(\nu-1)^{\nu-1}}{\nu^\nu},\qquad
                    c=\Delta_\ell+\frac{2}{\ell}C_{H,k},\qquad \nu=\ell+k-1,\\
                    C_{H,k}=2k(d_H+1)^{k-1+d_H/2}
                    \sum_{\substack{0\leq i\leq j\\i+j\leq d_H}}|\lambda_{i,j}|(j-i).
                \end{gathered}
            \end{equation*}
            Therefore, $-i[H,\cdot] + \cL_\ell$ generates a Sobolev regularizing QMS, which satisfies
            \begin{equation*}
                \left\|e^{t(-i[H,\cdot]+\cL_\ell)}\right\|_{1\to W^{k,1}}
                \leq\left(\frac{2k}{\ell(\ell-1)t}\right)^{k/(\ell-1)}
                +\bigl(c_{k,H}^{(\ell)}\bigr)^{k/(k+\ell-1)}.
            \end{equation*}

        \end{lem}
        \begin{proof}
            The proof uses the moment argument of Lemma 4.6 in \cite{Gondolf.2024}. For completeness, the coefficient bound above follows by writing the commutator with $(N+\1)^k$ in the Fock basis: a monomial shifts occupation by $j-i$, its row coefficient is bounded by $(n+d_H+1)^{d_H/2}$, and the weight difference by $k(j-i)(n+d_H+1)^{k-1}$. The absolute row sum therefore bounds $\pm i[H,(N+\1)^k]$ by $C_{H,k}(N+\1)^{k-1+d_H/2}$. Since $d_H\leq2(\ell-1)$, this is a lower-order term in the photon moment inequality. Scalar optimization gives the displayed constants. Generation and regularization then follow as in Lemma~\ref{lem:l-diss}.
        \end{proof}
        \begin{rmk}
            
            Let $J=a^{\ell-1}(a-\alpha)$, as in the learning scheme of \cite{Moebus.2025}. The finite-shift argument used in Lemmas~\ref{lem:l-diss} and~\ref{lem:l-diss-hamiltonian} shows that, for every real $k\geq1$, there is $c_{J,k}<\infty$ such that
            \begin{equation*}
                \cL[J]^*((N+\1)^k)
                \leq-\frac\ell2(N+\1)^{k+\ell-1}+c_{J,k}\1.
            \end{equation*}
            Thus the Sobolev regularization condition holds with gain $\ell-1$, and Theorem~\ref{thm:single-mode-generation-regularization} applies.
            
        \end{rmk}

    \subsection{Exponential convergence in $1\rightarrow1$ norm }\label{subsec:exp-convergence}
        Azouit, Sarlette, and Rouchon \cite[Thm.~2]{Azouit.2016} established convergence of the shifted $\ell$-photon dissipation in an adapted weighted trace-class norm, which controls the trace norm, for initial states in that space. They also proved the exponential Lyapunov estimate
        \begin{equation*}
            \tr\big[L_\ell e^{t\cL_\ell}(\rho)L_\ell^\dagger\big] \leq e^{-t\ell!} \tr[L_\ell\rho L_\ell^\dagger]\,,
        \end{equation*}
        where $L_\ell=a^\ell-\alpha^\ell$. Related quantitative estimates were considered in \cite[Prop.~A.2]{Moebus.2025}. Here we formulate convergence using the asymptotic completely positive, trace-preserving (CPTP) projection $\Pi_\ell$. This limiting channel is distinct from the orthogonal compression onto the code space. The argument below establishes the weighted estimate
        \begin{equation*}
            \begin{aligned}
                \|(e^{t\cL_\ell} - \Pi_\ell)(\rho)\|_1 &\leq 6e^{- \frac{t}{2} \ell!}\biggl( \Bigl(1 + \frac{|\alpha|^\ell}{\sqrt{\ell!}}\Bigr)^2\|\rho \|_{W^{\ell,1}} + \|\rho\|_1 \biggr)
            \end{aligned}
        \end{equation*}
        for all $\rho \in W^{\ell,1}$ and $\Pi_\ell$ is the asymptotic CPTP projection onto $\mathcal{C}_\ell(\alpha)$, satisfying $\Pi_\ell e^{t\cL_\ell}=e^{t\cL_\ell}\Pi_\ell=\Pi_\ell$. By a simple trick and using all the machinery above, we can prove the following result:
        \begin{lem}\label{lem:cat-convergence}
            Let $\ell \geq 2$ be an integer and $\Pi_\ell$ be the asymptotic CPTP projection onto $\mathcal{C}_\ell(\alpha)$ \eqref{eq:codespace-ell-photon-diss}. Then, for all $t>0$
            \begin{equation*}
                \begin{aligned}
                    \|e^{t\cL_\ell}-\Pi_\ell\|_{1\to1}
                    \leq6e^{-t\ell!/4}\biggl[\Bigl(1+\frac{|\alpha|^\ell}{\sqrt{\ell!}}\Bigr)^2
                    \biggl\{\Bigl(\frac4{(\ell-1)t}\Bigr)^{\ell/(\ell-1)}
                    +\bigl(c_\ell^{(\ell)}\bigr)^{\ell/(2\ell-1)}\biggr\}+1\biggr].
                \end{aligned}
            \end{equation*}
            where $c_\ell^{(\ell)}$ is given by Lemma~\ref{lem:l-diss}.
        \end{lem}
        \begin{proof}
            
            By Lemma~\ref{lem:core-drift-transfer}, the energy-decay estimate extends to positive inputs in $W^{\ell,1}$. Since $L_\ell^\dagger L_\ell\geq\ell!(\1-P_\ell)$, the compression inequality gives, for every state $\rho\in W^{\ell,1}$ and $s\geq0$,
            \begin{equation*}
                \|e^{(t+s)\cL_\ell}\rho-e^{t\cL_\ell}\rho\|_1
                \leq4e^{-t\ell!/2}\sqrt{\tr[L_\ell^\dagger L_\ell\rho]/\ell!}.
            \end{equation*}
            Hence the orbit converges. By density and contractivity, its limit defines a CPTP projection $\Pi_\ell$ onto $\mathcal C_\ell(\alpha)$ satisfying $\Pi_\ell e^{t\cL_\ell}=e^{t\cL_\ell}\Pi_\ell=\Pi_\ell$. The bound
            \begin{equation*}
                \frac{\tr[L_\ell^\dagger L_\ell\rho]}{\ell!}
                \leq\left(1+\frac{|\alpha|^\ell}{\sqrt{\ell!}}\right)^2
                \|\rho\|_{W^{\ell,1}}
            \end{equation*}
            and the standard $2\times2$ self-adjoint dilation extend the resulting estimate from states to all operators in $W^{\ell,1}$, proving the displayed weighted bound.
            
            First, we assume $\rho\in W^{\ell,1}$. Then the weighted estimate applied to $\rho_{t/2}=e^{t\cL_\ell/2}\rho$ shows
            \begin{equation*}
                \|(e^{\frac{t}{2}\cL_\ell} - \Pi_\ell)(\rho_{\frac{t}{2}})\|_1\leq 6e^{- \frac{t}{4} \ell!}\biggl( \Bigl(1 + \frac{|\alpha|^\ell}{\sqrt{\ell!}}\Bigr)^2\|\rho_{\frac{t}{2}}\|_{W^{\ell,1}} + \|\rho_{\frac{t}{2}}\|_1 \biggr)\,.
            \end{equation*}
            Here, we also used that $\rho\in W^{\ell,1}$ implies $e^{\frac{t}{2}\cL_\ell}(\rho)\in W^{\ell,1}$ due to Lemma~\ref{lem:l-diss}. Then, a final application of Lemma~\ref{lem:l-diss}, i.e.,
            \begin{equation*}
                \|e^{t\cL_\ell/2}\|_{1\to W^{\ell,1}}
                \leq\left(\frac4{(\ell-1)t}\right)^{\ell/(\ell-1)}
                +\bigl(c_\ell^{(\ell)}\bigr)^{\ell/(2\ell-1)}.
            \end{equation*}
            shows
            \begin{equation*}
                \begin{aligned}
                    \|(e^{t\cL_\ell}-\Pi_\ell)(\rho)\|_1
                    \leq6e^{-t\ell!/4}\biggl[\Bigl(1+\frac{|\alpha|^\ell}{\sqrt{\ell!}}\Bigr)^2
                    \biggl\{\Bigl(\frac4{(\ell-1)t}\Bigr)^{\ell/(\ell-1)}
                    +\bigl(c_\ell^{(\ell)}\bigr)^{\ell/(2\ell-1)}\biggr\}+1\biggr]\|\rho\|_1.
                \end{aligned}
            \end{equation*}
            The identity $\Pi_\ell e^{t\cL_\ell/2}=\Pi_\ell$ and density of $W^{\ell,1}$ finish the proof.
        \end{proof}

    \subsection{Perturbation theory in operator norm}\label{subsec:applicatio-perturbation-theory}
        With the results above, we are now able to directly apply the results of Section~\ref{subsec:perturbation-theory}, which address several aspects of perturbing the bosonic cat code. We begin with a standard result concerning finite-time perturbations.

        \begin{lem}\label{lem:application-finite-time-perturbation}
            Let $\ell \geq 2$ be an integer, $\alpha \in \mathbb{C}$, and let $H$ be the Hamiltonian \eqref{eq:single-mode-H} with $d_H \leq \ell - 2$.
            
            For $1\leq d_H\leq\ell-2$ and $\varepsilon\geq0$, put $c_1=2\|H(N+\1)^{-d_H/2}\|_\infty$. Then
            \begin{equation*}
                \begin{aligned}
                    \|e^{t\cL_\ell}-e^{t(\cL_\ell-i\varepsilon[H,\cdot])}\|_{1\to1}
                    \leq\varepsilon c_1\biggl[&\frac{\ell-1}{\ell-1-d_H}
                    \left(\frac{2d_H}{\ell(\ell-1)}\right)^{d_H/(\ell-1)}t^{1-d_H/(\ell-1)}\\
                    &+\bigl(c_{d_H}^{(\ell)}\bigr)^{d_H/(d_H+\ell-1)}t\biggr].
                \end{aligned}
            \end{equation*}
            For $d_H=0$, the Hamiltonian is scalar and the two semigroups coincide.
            
        \end{lem}
        \begin{proof}
            This is a direct application of Proposition~\ref{prop:intermediate-time-perturbation-bound}. The unperturbed quantitative regularization estimate \eqref{eq:quantitative-concentrated-regularization} follows from Lemma~\ref{lem:l-diss} with $\delta=\ell-1$, while Lemma~\ref{lem:l-diss-hamiltonian} shows that $-i\varepsilon[H,\cdot]+\cL_\ell$ generates a QMS and hence a trace-norm contraction semigroup. Furthermore, one may take $c_1=2\|H(N+\1)^{-d_H/2}\|_\infty<\infty$, since finite shifts bound the weighted polynomial coefficients and factorization through the Sobolev sandwich gives
            \begin{equation*}
                \|i[H,\rho]\|_1 \leq c_1 \|\rho\|_{W^{k,1}}\,,
            \end{equation*}
            with $k = d_H \leq \ell - 2$. Substituting these bounds into Proposition~\ref{prop:intermediate-time-perturbation-bound} completes the proof. A sufficiently high Sobolev space satisfies the \hyperref[assum:perturbation-common-domain]{common domain assumption}, by Lemma~\ref{lem:core-drift-transfer}.
        \end{proof}

        Next, we extend the large-time perturbation bound proven in Theorem 5.2 of \cite{Gondolf.2024} using the tools established above.

        \begin{lem}\label{lem:application-large-time-perturbation}
            Let $\ell \geq 2$ be an integer, $\alpha \in \mathbb{C}$, and let $H$ be the Hamiltonian \eqref{eq:single-mode-H} with $d_H \leq \ell - 2$.
            
            For $1 \leq d_H \leq \ell-2$ and $\varepsilon \geq 0$, put $\cE = -i[H,\cdot]$, $\Pi_\ell^\perp = \mathrm{id}-\Pi_\ell$, and $c_1 = 2\|H(N+\1)^{-d_H/2}\|_\infty$. For states $\rho,\sigma\in\cT_1$,
            \begin{equation*}
                \begin{aligned}
                    \|\Pi_\ell^\perp e^{t\cL_\ell}\sigma-
                    \Pi_\ell^\perp e^{t(\cL_\ell+\varepsilon\cE)}\rho\|_1
                    \leq\begin{cases}
                        \|\Pi_\ell^\perp(\sigma-\rho)\|_1+
                        2\varepsilon t^{1-d_H/(\ell-1)}\hat C_1\|\rho\|_1,&t\leq1,\\[4pt]
                        \widetilde C e^{-t\ell!/4}\|\Pi_\ell^\perp(\sigma-\rho)\|_1+
                        \varepsilon\hat C_{\rm long}^{\perp}\|\rho\|_1,&t\geq1.
                    \end{cases}
                \end{aligned}
            \end{equation*}
            The constants are those of Proposition~\ref{prop:invariant-subset-perturbation-bound} with $k=d_H$, $\delta=\ell-1$, $c_2=0$, $\gamma=\ell!/4$, and with $c_{d_H,\varepsilon H}^{(\ell)}$ denoting $c_{k,H}^{(\ell)}$ from Lemma~\ref{lem:l-diss-hamiltonian} evaluated at $k=d_H$ and $H$ replaced by $\varepsilon H$:
            \begin{equation*}
                \begin{gathered}
                    \mu_k=\ell/2,\qquad c_k=(\ell/2)c_{d_H,\varepsilon H}^{(\ell)},\\
                    \widetilde C=6\left[\left(1+\frac{|\alpha|^\ell}{\sqrt{\ell!}}\right)^2
                    \left\{\left(\frac4{\ell-1}\right)^{\ell/(\ell-1)}
                    +\bigl(c_\ell^{(\ell)}\bigr)^{\ell/(2\ell-1)}\right\}+1\right].
                \end{gathered}
            \end{equation*}
            These constants are uniform for $\varepsilon$ in a fixed bounded interval of $[0,\infty)$. The scalar case $d_H=0$ has no perturbative contribution.
            
        \end{lem}
        \begin{proof}
            The result follows by direct application of Proposition~\ref{prop:invariant-subset-perturbation-bound}. The required assumptions hold due to Lemma~\ref{lem:l-diss-hamiltonian}, which ensures that $-i\varepsilon[H,\cdot]+\cL_\ell$ generates a QMS satisfying \eqref{eq:quantitative-concentrated-regularization} with $\delta=\ell-1$, and by Lemma~\ref{lem:cat-convergence}, which gives
            
            \begin{equation*}
                \|e^{t\cL_\ell}-\Pi_\ell\|_{1\to1}
                \leq\widetilde C e^{-t\ell!/4},\qquad t\geq1.
            \end{equation*}
            The relative bound and the \hyperref[assum:perturbation-common-domain]{common domain assumption} are verified as in the preceding lemma, while the constants $c_{d_H,\varepsilon H}^{(\ell)}$ regularize the perturbed orbit. Their explicit formula is bounded on bounded $\varepsilon$ intervals.
            
            Substituting these expressions into Proposition~\ref{prop:invariant-subset-perturbation-bound} finishes the proof.
        \end{proof}

        \begin{rmk}
            In both results above, other time-independent Markovian noise can also be accommodated under the \hyperref[assum:perturbation-common-domain]{common domain assumption} and the generation and relative Sobolev assumptions. Time-dependent or non-Markovian noise requires separate assumptions on an evolution family and its Duhamel formula.
        \end{rmk}

        \begin{rmk}
            Following the construction in \cite{Moebus.2025}, namely
            \begin{equation*}
                \cL[a - \alpha] + \cL[a^\ell - \alpha^\ell],
            \end{equation*}
            
            the state $\ketbra\alpha\alpha$ is a common stationary state, and the same generation and regularization method applies. To invoke Proposition~\ref{prop:steady-state-perturbation-bound}, a unique-state mixing estimate in the precise $1\to1$ topology is additionally required. Generation and a common stationary state alone do not establish it.
            
        \end{rmk}

\section{Multi-mode applications}\label{sec:multi-mode-applications}
    In this section, we extend the single-mode examples presented in Section~\ref{sec:single-mode-applications} to multi-mode systems and obtain bounds for local Hamiltonian perturbations. One motivation for this analysis is the study of multi-mode bosonic error-correcting cat-code extensions, which are discussed, for example, in \cite{Jain.2024}. Although we focus on photon-loss dissipation, many other forms of dissipation should satisfy similar properties when analyzed using the arguments developed here.

    \noindent As in the previous sections, detailed calculations are provided in Appendix~\ref{appx:technical-bounds-application} for the sake of clarity.

    \subsection{Multi-mode photon dissipation}\label{subsec:multi-mode-photon-dissipation}
        In this section, we show $k$-concentrated Sobolev regularization for multi-mode $\ell$-photon dissipation in the presence of nearest-neighbor Hamiltonians consisting of polynomials in creation and annihilation operators. The interaction structure of the model is given by a $D$-dimensional lattice defined by a graph $(V,E)$ with $V \subset \mathbb{Z}^D$ and $|V| = m$. We include the on-site terms as loops in $E$. We work on connected finite volumes, with maximum number of incident terms bounded by $\Delta$, and assume uniform bounds on the polynomial coefficients and on $\alpha_*\coloneqq\max_i|\alpha_i|$. Disconnected volumes can be treated componentwise. The dissipation is taken to be the sum of (shifted) photon-dissipation terms acting on single-modes, i.e.
        \begin{equation}\label{eq:multi-mode-ell-dissipation}
            \cL_\ell(\rho) = \sum_{j \in V} \cL[a_j^\ell - \alpha_j^\ell](\rho) \, ,
        \end{equation}
        for a vector $\alpha \in \mathbb{C}^m$. Recall that $\cL[L] = L \cdot L^\dagger - \tfrac{1}{2}\{ L^\dagger L , \cdot \}$ for any operator $L$. Alongside this Lindbladian, we also consider a $2$-local Hamiltonian of the form
        \begin{equation}\label{eq:2-body-Hamiltonian}
            H = \sum_{i \sim j} H_{ij}= \sum_{i \sim j} \sum_{\substack{u \in \mathbb{N}_0^4 \\ \|u\|_1 \le d_H}}\bigl(\lambda_{ij}^{(u)}\, a_i^{u_1} a_j^{u_3} (a_i^\dagger)^{u_2} (a_j^\dagger)^{u_4}+ \overline{\lambda}_{ij}^{(u)}\, a_j^{u_4} a_i^{u_2} (a_j^\dagger)^{u_3} (a_i^\dagger)^{u_1}\bigr),
        \end{equation}
        where $i \sim j$ means $(i,j) \in E$, corresponding to nearest neighbors on the $D$-dimensional lattice. Together, this defines the operator
        \begin{equation}\label{eq:hamiltonian-dissipation}
            \cL_\ell^{(H)}\coloneqq-i[H,\cdot]+\cL_\ell\,.
        \end{equation}
        The following theorem establishes the concentrated Sobolev regularization condition and its consequences for the full dynamics.

        \begin{thm}\label{thm:multimode_photon_diss}
            Let $\cL_\ell$ be a Lindbladian of the form \eqref{eq:multi-mode-ell-dissipation} associated to $\ell$-photon-dissipation and $H$ a $2$-local Hamiltonian of the form \eqref{eq:2-body-Hamiltonian}.
            
            Suppose that $\ell \geq 2$ is an integer and $d_H\le2(\ell-1)$. For every fixed $\eta>0$ and every real $k\ge1$, there is $c_k<\infty$ such that
            \begin{equation}\label{eq:mutli-mode-sobo-stability1}
                \tr[M_{v,k}(-i[H,\rho]+\eta\cL_\ell(\rho))]
                \le-\mu\tr[M_{v,k+\ell-1}\rho]+c_k\tr\rho,
                \quad \rho\in\cT_f,\quad\rho\ge0,
                \qquad\mu=\frac{\eta\ell}{2}.
            \end{equation}
            The constants are independent of the center and the number of modes under the uniform coefficient and lattice assumptions above. The closure of the pre-generator generates a QMS, with $\cT_f$ a generator core on trace class and on each $W_v^{k,1}$, $k\ge1$. Moreover,
            \begin{equation}\label{eq:mutli-mode-sobo-stability2}
                \|e^{t(-i[H,\cdot]+\eta\cL_\ell)}\|_{1\to W_{\max}^{k,1}}
                \le\left(\frac{k}{(\ell-1)\mu t}\right)^{k/(\ell-1)}
                +\left(\frac{c_k}{\mu}\right)^{k/(k+\ell-1)},\qquad t>0.
            \end{equation}
            The same fixed $\eta$ works for all moment orders.
            
        \end{thm}
        \begin{proof}
            Proposition~\ref{appx-thm:multimode_photon_diss} proves \eqref{eq:mutli-mode-sobo-stability1} for every real $k \ge 1$, uniformly in the center and volume. This implies the Sobolev stability condition at every such order, so Corollary~\ref{cor:concentrated-generation-unbounded} gives generation and the generator-core assertions. Theorem~\ref{thm:concentrated-generation-regularization} then yields \eqref{eq:mutli-mode-sobo-stability2}.
        \end{proof}

    \subsection{Local Hamiltonian perturbations}
        We derive volume-independent bounds on changes in observable expectations under Hamiltonian perturbations supported on a fixed finite set. The bound does not assert decay with the distance between that set and the observable.

        \begin{cor}\label{cor:local-polynomial-perturbation}
            
            Under the assumptions of Theorem~\ref{thm:multimode_photon_diss}, let $H_B$ be a $2$-local Hamiltonian of the form \eqref{eq:2-body-Hamiltonian}, supported on a fixed finite set $B \subset V$, of total polynomial degree $1 \le r < 2(\ell - 1)$. Put $\delta = \ell - 1$ and $\mu = \ell/2$. For every $\varepsilon \in \mathbb R$, the closures of $\cL_\ell^{(H)}$ and $\cL_\ell^{(H)} - i\varepsilon[H_B,\cdot]$ generate QMS, denoted by $T_t$ and $T_t^\varepsilon$. For every $t \ge 0$, state $\rho$, and bounded observable $O_A$ supported on $A \subset V$,
            \begin{equation}
                \begin{split}
                    \left|\tr[O_A(T_t^\varepsilon-T_t)(\rho)]\right|
                    \le |\varepsilon|K_B\|O_A\|_\infty\Bigg[
                    \frac{(r/(\delta\mu))^{r/(2\delta)}}{1-r/(2\delta)}
                    t^{1-r/(2\delta)}
                    +\left(\frac{c_r}{\mu}\right)^{r/(2(r+\delta))}t
                    \Bigg],
                \end{split}
            \end{equation}
            where $c_r$ is the unperturbed regularization constant in Theorem~\ref{thm:multimode_photon_diss} with $\eta = 1$. The constant $K_B$ depends on $H_B$, $B$, and the fixed concentration parameters, but not on the ambient volume.
            
        \end{cor}
        \begin{proof}
            
            The finite Fock shifts of $H_B$ and the scalar arithmetic--geometric mean inequality give $\|H_B\psi\|^2 \le C_B\langle\psi,\sum_{i\in B}(N_i + \1)^r\psi\rangle$ for $\psi \in \mathcal H_f$. Hence, for a state $\sigma$ with finite local moments of order $r$ on $B$, trace Cauchy--Schwarz gives
            \begin{equation}
                \|[H_B,\sigma]\|_1\le2\left(C_B\sum_{i\in B}\tr[(N_i+\1)^r\sigma]\right)^{1/2}.
            \end{equation}
            Applying the regularization estimate at center $i \in B$ and using $Z_i \le Z_*(\kappa)$ from Lemma~\ref{appx-lem:normalization-bounds} yields
            \begin{equation}
                \tr[(N_i+\1)^rT_s\rho]
                \le Z_*(\kappa)\left[\left(\frac{r}{\delta\mu s}\right)^{r/\delta}
                +\left(\frac{c_r}{\mu}\right)^{r/(r+\delta)}\right],\qquad s > 0.
            \end{equation}
            By Theorem~\ref{thm:multimode_photon_diss} and Lemma~\ref{lem:core-drift-transfer}, a sufficiently high Sobolev space satisfies the \hyperref[assum:perturbation-common-domain]{common domain assumption}. Lemma~\ref{lem:perturbation-duhamel} therefore gives, for every state $\rho \in \cT_f$,
            \begin{equation}
                (T_t^\varepsilon-T_t)\rho
                =\varepsilon\int_0^tT_{t-s}^\varepsilon(-i[H_B,T_s\rho])\,ds.
            \end{equation}
            Contractivity, the preceding moment estimate, and $\sqrt{a+b}\le\sqrt a+\sqrt b$ give the stated bound with $K_B=2(C_B|B|Z_*(\kappa))^{1/2}$. The integral at zero is finite because $r<2\delta$. Positive Fock approximation extends the resulting estimate to every state. Finally, use trace duality with $O_A$.
            
        \end{proof}

\section{Conclusion}\label{sec:conclusion}
    In this work, we identified a broad class of unbounded GKSL generators on bosonic Fock space whose associated quantum Markov semigroups exhibit instantaneous Sobolev regularization. We showed that dissipative evolutions generated by polynomial creation and annihilation operators not only preserve finite moments but also immediately improve them, mapping arbitrary trace-class states into highly regular ones for every $t>0$. This provides a structural explanation for stability phenomena in bosonic continuous-variable systems and yields explicit convergence bounds for arbitrary input states. In particular, our bounds sharpen existing perturbative results at both short and long times, offering new analytic tools for assessing robustness in bosonic quantum information processing.

    The single-sandwich generation theorem supplies the semigroup and moment-domain properties from concentrated Sobolev stability conditions at two suitably separated orders. Uniform local moment constants are obtained for finite volumes.

    The multi-mode framework developed here points toward several promising directions. A central open question is whether instantaneous Sobolev regularization can be exploited to improve bosonic information-propagation bounds, such as Lieb–Robinson–type estimates. Another natural direction for future work is the explicit analysis of particular bosonic quantum error-correcting codes and the implications of our theory for their performance.

    \vspace{2ex}
    \emph{Acknowledgments} We would like to thank Robert Salzmann, Cambyse Rouzé, and Marius Lemm for their valuable discussions on the topic. This work was funded by the Deutsche Forschungsgemeinschaft (DFG, German Research Foundation) - Project-ID 470903074 - TRR 352 (PG, TM), Project-ID 575156903 (TM), and  Germany’s Excellence Strategy-EXC2111-390814868 (PCR), QuantERA II Programme that has received funding from the EU’s H2020 research and innovation programme under the GA No 101017733 (TM).

\section*{Statements and Declarations}
    \noindent\textbf{Competing interests.} The authors declare no competing financial or non-financial interests.

    \noindent\textbf{Data availability.} No datasets were generated or analysed in this theoretical study.

    \noindent\textbf{Use of AI assistance.} The authors developed the main mathematical ideas entirely through discussions and worked out the main proofs at the blackboard. An OpenAI language-model assistant was used to assist with proof development, mathematical consistency checks, and manuscript revision. Responsibility for the mathematical claims and the final manuscript rests with the authors.

\bibliographystyle{abbrvnat}
\bibliography{bibliography_v3}

\appendix
\addtocontents{toc}{\protect\setcounter{tocdepth}{1}}
\addcontentsline{toc}{section}{Appendices}
\addtocontents{toc}{\protect\setcounter{tocdepth}{0}}

\section{ODE comparison lemma}\label{appx:ode-comparison-lemma}
    In the following section, we prove Lemma~\ref{lem:gronwall-extension}. For clarity of presentation, we begin by restating the lemma:
    \begin{lem}\label{appx-lem:gronwall-extension}
        Let $y : [0, \infty) \to [0, \infty)$ be a locally absolutely continuous function satisfying the differential inequality almost everywhere
        \begin{equation*}
            \frac{dy}{dt}(t) \le -a\, y(t)^p + b
        \end{equation*}
        for constants $a>0$, $b\geq0$ and exponent $p > 1$. Then,
        
        \begin{equation*}
            y(t)\leq z(t):=(b/a)^{1/p}+
            \begin{cases}
                ((y(0)-(b/a)^{1/p})^{1-p}+(p-1)at)^{-1/(p-1)},&y(0)>(b/a)^{1/p},\\
                0,&y(0)\leq(b/a)^{1/p}.
            \end{cases}
        \end{equation*}

    \end{lem}
    \begin{proof}
        In a first step, we define
        \begin{equation*}
            t_c \coloneqq \inf\{t \ge 0 : y(t) \le (b/a)^{1/p}\}\,.
        \end{equation*}
        For $t > t_c$, it follows immediately that $y(t) \le (b/a)^{1/p}$, and hence $y(t) \le z(t)$. We verify this by contradiction. Suppose there exists $t_v > t_c$ such that $y(t_v) > (b/a)^{1/p}$. By continuity of $y$, the intermediate value theorem implies that there exist $t_0, t_1$ with $t_0 < t_1$ such that
        \begin{equation*}
            y(t_0) = (b/a)^{1/p}, \qquad y(t_1) > (b/a)^{1/p}, \qquad \text{and} \quad y(t) \ge (b/a)^{1/p} \quad\text{ for all } t \in [t_0, t_1]\,.
        \end{equation*}
        
        Since $y(t)\geq(b/a)^{1/p}$ on $[t_0,t_1]$, the differential inequality gives $y'(t)\leq0$ almost everywhere on this interval. Absolute continuity therefore implies
        \begin{equation*}
            y(t_1)-y(t_0)=\int_{t_0}^{t_1}y'(s)\,ds\leq0,
        \end{equation*}
        
        a contradiction. Therefore, $y(t) \le (b/a)^{1/p}$ for all $t \ge t_c$.

        For $0 \le t < t_c$, define $g(t) \coloneqq y(t) - (b/a)^{1/p} > 0$. Then $g'(t) = y'(t)$ almost everywhere, and the differential inequality yields
        \begin{equation*}
            g'(t) \le -a\, g(t)^p.
        \end{equation*}
        Separating variables and integrating from $0$ to $t$ gives
        \begin{equation*}
            \frac{g(t)^{1 - p} - g(0)^{1 - p}}{1 - p} = \int_0^t \frac{g'(s)}{g(s)^p}\, ds \le \int_0^t -a\, ds = -a t.
        \end{equation*}
        Rearranging terms, we obtain
        \begin{equation*}
            g(t)^{1 - p} \ge g(0)^{1 - p} + (p - 1)a t.
        \end{equation*}
        Taking both sides to the power of $\tfrac{1}{1 - p}$ (which reverses the inequality since $\tfrac{1}{1 - p} < 0$), we find
        \begin{equation*}
            g(t) \le \frac{1}{\big(g(0)^{1 - p} + (p - 1)a t \big)^{\frac{1}{p - 1}}}.
        \end{equation*}
        Substituting back $g(t) = y(t) - (b/a)^{1/p}$ yields
        \begin{equation*}
            y(t) \le \frac{1}{\big((y(0) - (b/a)^{1/p})^{1 - p} + (p - 1)a t \big)^{\frac{1}{p - 1}}} + (b/a)^{1/p}.
        \end{equation*}
        Together with the case $t\geq t_c$, this completes the proof, including $b=0$.
    \end{proof}

\section{Multi-mode regularization and applications}\label{appx:multi-mode-extension}
    \subsection{Concentrated moments and regularization}

        In this appendix, we restate the multi-mode extensions previously introduced in Section~\ref{subsec:multi-mode-regularizing} and provide detailed proofs. We first record the volume-uniform bound on the normalization constants.
        \begin{lem}\label{appx-lem:normalization-bounds}
            Let $\mathfrak V$ be a family of finite lattices for which there is a function $g_D:\N_0\to[0,\infty)$ satisfying
            \begin{equation*}
                \sup_{V\in\mathfrak V}\sup_{v\in V}
                |\{i\in V:\dist(v,i)=\ell\}|
                \leq g_D(\ell),
                \qquad
                \sum_{\ell=0}^{\infty}g_D(\ell)e^{-\kappa\ell}<\infty .
            \end{equation*}
            Then
            \begin{equation}\label{appx-eq:uniform-normalization-bound}
                1\leq Z_v\leq
                Z_*(\kappa)\coloneqq
                \sum_{\ell=0}^{\infty}g_D(\ell)e^{-\kappa\ell},
                \qquad V\in\mathfrak V,\quad v\in V.
            \end{equation}
            In particular, for finite boxes of a fixed $D$-dimensional lattice one may take
            \begin{equation*}
                g_D(\ell)=
                2^D\frac{(D+\ell)^{D-1}}{(D-1)!}.
            \end{equation*}
            Thus $Z_*(\kappa)<\infty$ for every $\kappa>0$. For $\kappa>1$, the convenient rough estimate
            \begin{equation*}
                Z_*(\kappa)\leq
                \frac{e^{2D}}{1-e^{-(\kappa-1)}}
            \end{equation*}
            is available.
        \end{lem}
        \begin{proof}
            Decomposing the sum into distance spheres and using the assumed sphere bound gives
            \begin{equation*}
                Z_v
                =\sum_{\ell\geq0}
                |\{i\in V:\dist(v,i)=\ell\}|e^{-\kappa\ell}
                \leq Z_*(\kappa).
            \end{equation*}
            The term $i=v$ gives $Z_v\geq1$. The stated polynomial sphere estimate is standard (see, for example, \cite[Eq.~(B2)]{Moebus.2025Stability}). Its series converges for every $\kappa>0$. Bounding the polynomial by an exponential yields the final rough estimate when $\kappa>1$.
        \end{proof}
        Next, we prove the following inequality, which motivates the definition of concentrated Sobolev norms --- local moments bound the global but concentrated Sobolev norm.
        \begin{lem}\label{appx-lem:bound-coherent-state}
            Let $\alpha\in\C^m$ define the coherent state $\ket{\alpha}$, then for $k>0$ and $\theta>0$, writing $\alpha_*=\max_i|\alpha_i|$,
            \begin{align*}
                \|\ketbra{\alpha}{\alpha}\|_{W^{k,1}_v}
                &\leq\max_{i\in\{1,...,m\}}\tr[\ketbra{\alpha_i}{\alpha_i}(N_i+\1)^k]\\
                &\leq\left(\frac{k}{e\theta}\right)^k
                \exp\bigl(\theta+\alpha_*^2(e^\theta-1)\bigr).
            \end{align*}
        \end{lem}
        \begin{proof}
            By definition of the concentrated Sobolev norm, the norm reduces to
            \begin{equation*}
                \|\ket{\alpha}\bra{\alpha}\|_{W^{k,1}_v}=\sum_{i=1}^m\frac{e^{-\kappa\dist(v,i)}}{Z_v}\bra{\alpha_i}(N_i+ \1)^k\ket{\alpha_i}\,,
            \end{equation*}
            and by that to a single mode. For a Poisson random variable $X$ of mean $|\alpha_i|^2$, the scalar inequality $x^k\leq(k/(e\theta))^k e^{\theta x}$ gives
            
            \begin{equation}
                \bra{\alpha_i}(N_i+\1)^k\ket{\alpha_i}
                =\mathbb E(X+1)^k
                \leq\left(\frac{k}{e\theta}\right)^k
                e^{\theta+|\alpha_i|^2(e^\theta-1)}.
            \end{equation}
            The normalized positive weights give the claim. This estimate includes zero amplitudes and all real $k>0$.
            
        \end{proof}

        We begin with a Jensen-type inequality for multi-mode bosonic moments:
        \begin{lem}[Jensen’s inequality for multi-mode moments]\label{appx-lem:jensens-inequality-moments-number-operator-multi}
            Fix $v\in V$. Let $p\geq q>0$ and let $\rho\in W_v^{p,1}$ be a state. Then
            \begin{equation}\label{appx-eq:multi-jensen-moment}
                \tr[\cW_v^{p} (\rho)]\geq  \big( \tr[\cW_v^{q} (\rho)] \big)^{\frac{p}{q}}\,.
            \end{equation}
        \end{lem}
        \begin{proof}
            Let
            \begin{equation*}
                P_{i,n}\coloneqq
                \ketbra{n}{n}_i\otimes\1_{V\setminus\{i\}},
                \qquad
                p_{i,n}\coloneqq
                \alpha_{v,i}\tr[\rho P_{i,n}].
            \end{equation*}
            Then $(p_{i,n})_{i\in V,n\geq0}$ is a probability distribution and its $p$-moment is finite. Moreover,
            \begin{equation*}
                \tr[\cW_v^s(\rho)]
                =\sum_{i\in V}\sum_{n\geq0}
                p_{i,n}(n+1)^s,
                \qquad s>0.
            \end{equation*}
            Jensen's inequality for $u\mapsto u^{p/q}$ gives
            \begin{align*}
                \tr[\cW_v^p(\rho)]
                &=\sum_{i,n}p_{i,n}
                \bigl((n+1)^q\bigr)^{p/q}\\
                &\geq
                \left(\sum_{i,n}p_{i,n}(n+1)^q\right)^{p/q}
                =\bigl(\tr[\cW_v^q(\rho)]\bigr)^{p/q}.
            \end{align*}
        \end{proof}
        \noindent The fixed-center estimate underlying Theorem~\ref{thm:sobolev-regularization-multi-mode} is the following.
        \begin{prop}[Fixed-center concentrated regularization]
            \label{appx-thm:sobolev-regularization-multi-mode} Under the hypotheses of Theorem~\ref{thm:sobolev-regularization-multi-mode} at a fixed center $v \in V$, for all $t > 0$,
            \begin{equation}\label{appx-eq:fixed-center-regularization-bound}
                \|T_t\|_{1\to W_v^{k,1}}
                \leq
                B_k(t)\coloneqq
                \left(\frac{k}{\delta\mu_k t}\right)^{k/\delta}
                +\left(\frac{c_k}{\mu_k}\right)^{k/(k+\delta)}.
            \end{equation}
        \end{prop}
        \begin{proof}
            For a positive finite-Fock-support state $\rho$, put $y(t)=\tr[\cW_v^k(T_t(\rho))]$. By Lemma~\ref{appx-lem:jensens-inequality-moments-number-operator-multi},
            \begin{equation*}
                y'(t)
                \leq-\mu_k y(t)^{1+\delta/k}+c_k
            \end{equation*}
            almost everywhere. The scalar comparison result in Lemma~\ref{lem:gronwall-extension} gives $y(t)\leq B_k(t)$. Positive finite-Fock cutoffs and lower semicontinuity extend this estimate to every state.

            The positive moment bound ensures that $M_{v,k}^{1/2}T_t(x)M_{v,k}^{1/2}$ is trace class for positive $x$, and hence for every $x \in \cT_1$ by linear decomposition. The $2\times2$ argument in \cite[proof of Thm.~4.1.1]{Moebus.2025thesis}, applied with input weight $\1$ and output weight $M_{v,k}$, gives
            \begin{equation*}
                \|T_t(x)\|_{W_v^{k,1}} \le B_k(t)\|x\|_1,
                \qquad x \in \cT_1,
            \end{equation*}
            which finishes the proof of the statement.
            
        \end{proof}

    \subsection{Moment estimates for the multi-mode applications}\label{appx:technical-bounds-application}
    
        Before diving into the proofs of different applications, we first recall, for completeness, some standard tools used to establish the following moment stability bounds (see \cite{Gondolf.2024, Moebus.2024Learning, Moebus.2025thesis}).

        We use the one-mode multiphoton-loss estimate of \cite[Lem.~4.3]{Gondolf.2024}, with the moment order translated to the convention $(N+\1)^k$ used here. The additional point needed for the interacting model is the following symmetric form estimate. Operator forms in this appendix are initially defined on $\mathcal H_f$, and the corresponding moment estimates are tested on positive elements of $\cT_f$.

        \begin{lem}\label{appx-lem:polynomial-commutator-form}
            Let $P=\prod_{j\in X}(a_j^\dagger)^{u_j}a_j^{w_j}$ be a monomial of total degree $d=\sum_j(u_j+w_j)$, supported on one or two sites, and let $K=\lambda P+\overline\lambda P^\dagger$. For $k\ge1$,
            \begin{equation}
                \pm i[K,(N_i+\1)^k]
                \le C_{k,d}|\lambda|\sum_{j\in X}(N_j+\1)^{k-1+d/2},
                \qquad i\in X,
            \end{equation}
            as quadratic forms. If $u_i=w_i$, the commutator is zero. Consequently,
            \begin{equation}\label{eq:weighted-polynomial-commutator}
                \pm i[K,M_{v,k}]
                \le C_{k,d}|\lambda|
                \left(\sum_{i\in X}\frac{e^{-\kappa\dist(v,i)}}{Z_v}\right)
                \sum_{j\in X}(N_j+\1)^{k-1+d/2}.
            \end{equation}
        \end{lem}
        \begin{proof}
            If $u_i=w_i$ the commutator vanishes. Otherwise $d\ge1$. Put $h_j=u_j-w_j$ and $d_j=u_j+w_j$. On the Fock basis, $P|n\rangle=p(n)|n+h\rangle$, where the coefficient vanishes on forbidden transitions and otherwise satisfies
            \begin{equation}
                |p(n)|\le C_d\prod_{j\in X}(n_j+1)^{d_j/2}.
            \end{equation}
            This follows directly from the creation and annihilation amplitudes. On every allowed transition, $n_j+1$ and $n_j+h_j+1$ differ by a bounded factor depending only on $d$. The mean value theorem therefore gives
            \begin{equation}
                |(n_i+h_i+1)^k-(n_i+1)^k|
                \le C_{k,d}(n_i+1)^{k-1}.
            \end{equation}
            Define
            \begin{equation*}
                F_i(n)=(n_i+1)^{k-1}\prod_{j\in X}(n_j+1)^{d_j/2}.
            \end{equation*}
            For $\psi\in\mathcal H_f$, write $\psi_n=\langle n,\psi\rangle$. The coefficient bounds above and $2|\psi_n\psi_{n+h}|\leq|\psi_n|^2+|\psi_{n+h}|^2$ give
            \begin{equation*}
                \begin{aligned}
                    |\langle\psi,i[K,(N_i+\1)^k]\psi\rangle|
                    &\leq C_{k,d}|\lambda|\sum_{n:\,p(n)\ne0}F_i(n)
                    (|\psi_n|^2+|\psi_{n+h}|^2)\\
                    &\leq C'_{k,d}|\lambda|\sum_n F_i(n)|\psi_n|^2.
                \end{aligned}
            \end{equation*}
            For the second inequality, reindex the terms containing $|\psi_{n+h}|^2$ and use $F_i(n)\leq C_{k,d}F_i(n+h)$ on allowed transitions. To bound the resulting diagonal weight, set
            \begin{equation*}
                \beta_j=\frac{d_j}{2}+(k-1)\mathbf{1}_{\{j=i\}},
                \qquad s=\sum_{j\in X}\beta_j=k-1+\frac d2>0.
            \end{equation*}
            The strict inequality follows from $h_i\ne0$, which implies $d\geq1$. The weighted arithmetic--geometric mean inequality yields
            \begin{equation*}
                \begin{aligned}
                    F_i(n)
                    &=\prod_{j\in X}\bigl((n_j+1)^s\bigr)^{\beta_j/s}\\
                    &\leq\sum_{j\in X}\frac{\beta_j}{s}(n_j+1)^s
                    \leq\sum_{j\in X}(n_j+1)^s.
                \end{aligned}
            \end{equation*}
            This bounds both signs of the commutator by the stated diagonal form. Multiplying by $\alpha_{v,i}$ and summing over $i\in X$ gives \eqref{eq:weighted-polynomial-commutator}; the commutators vanish for $i\notin X$.
        \end{proof}

        \begin{prop}\label{appx-thm:multimode_photon_diss}
            Under the assumptions of Theorem~\ref{thm:multimode_photon_diss}, for every fixed $\eta>0$ and real $k\ge1$ there is a finite volume-independent $B_k$ such that
            \begin{equation}\label{eq:multimode-preabsorption-core}
                (-i[H,\cdot]+\eta\cL_\ell)^*(M_{v,k})
                \le-\eta\ell M_{v,k+\ell-1}+B_kM_{v,k+\ell-2}.
            \end{equation}
            In particular, \eqref{eq:mutli-mode-sobo-stability1} holds with $\mu=\eta\ell/2$ and
            \begin{equation}\label{eq:multimode-core-constant}
                c_k=B_k\left(\frac{B_k}{\mu}\right)^{k+\ell-2}.
            \end{equation}
        \end{prop}
        \begin{proof}
            For the dissipative part, write $\widetilde N=N+\1$ in one mode and $\widetilde N_i=N_i+\1$ at each site. The calculation underlying \cite[Lem.~4.3]{Gondolf.2024} gives
            \begin{equation}\label{eq:pure-loss-core-coercivity}
                \cL[a^\ell]^*(\widetilde N^k)
                \le-\ell \widetilde N^{k+\ell-1}
                +\frac{\ell^2(\ell+1)}2 \widetilde N^{k+\ell-2}.
            \end{equation}
            Writing $n[-\ell+1:0] = (n-\ell+1)\cdots n$ for the falling factorial, the diagonal expression at $n\ge\ell$ is $-n[-\ell+1:0][(n+1)^k-(n-\ell+1)^k]$. Use $(n+1)^k-(n-\ell+1)^k\ge\ell(n+1)^{k-1}$ and $n[-\ell+1:0]\ge(n+1)^\ell-\ell(\ell+1)(n+1)^{\ell-1}/2$. At $n<\ell$ the dissipative expression is zero and the proposed upper bound is nonnegative. This also verifies the constants in our moment convention. For the shifted jump, direct expansion gives
            \begin{equation}
                \cL[a^\ell-c]
                =\cL[a^\ell]-i[H_c,\cdot],\qquad
                H_c=\frac{\overline c\,a^\ell-c(a^\dagger)^\ell}{2i}.
            \end{equation}
            Its moment contribution is bounded by the preceding lemma at order $k-1+\ell/2\le k+\ell-2$.

            Normal-ordering the monomials of $H$ does not increase their degree. Lemma~\ref{appx-lem:polynomial-commutator-form} applies to all one- and two-site monomials and bounds their contributions by moments of order $k - 1 + d_H/2 \le k + \ell - 2$. Thus, for every term $H_{ij}$, including $i = j$,
            \begin{equation}
                \pm i[H_{ij}, M_{v,k}]
                \le C_k \Lambda (\alpha_{v,i} + \alpha_{v,j})
                (\widetilde N_i^{k+\ell-2} + \widetilde N_j^{k+\ell-2}),\qquad
                \alpha_{v,i} = \frac{e^{-\kappa\dist(v,i)}}{Z_v},
            \end{equation}
            where $\Lambda$ is a uniform bound on the coefficients in \eqref{eq:2-body-Hamiltonian}. The finite number of monomials of the prescribed degree is included in $C_k$. Summing over all edges, including loops, uses the incidence bound $\Delta$ and
            \begin{equation}
                \sum_{i:i\sim j}\alpha_{v,i}
                \le\Delta e^\kappa\alpha_{v,j}.
            \end{equation}
            Together with the shifted-loss estimate, this proves \eqref{eq:multimode-preabsorption-core}, for example with
            \begin{equation}
                B_k=\eta\frac{\ell^2(\ell+1)}2
                +C_{k,\ell,d_H}
                \left[\Lambda\Delta(1+e^\kappa)
                +\eta\alpha_*^\ell\right].
            \end{equation}
            No volume factor occurs in this reindexing.  Finally, for $a=k+\ell-1$ and $b=a-1$, the elementary scalar maximum satisfies $\sup_{x\ge0}(-\mu x^a+B_kx^b)\le B_k(B_k/\mu)^b$. Split the negative term of \eqref{eq:multimode-preabsorption-core} into two equal parts, apply this inequality at each site, and sum with the normalized weights. The scalar remainder is $c_k\sum_i\alpha_{v,i}=c_k$. This proves \eqref{eq:mutli-mode-sobo-stability1} with \eqref{eq:multimode-core-constant}.

        \end{proof}

\section{Perturbation theory}\label{appx:perturbation theory}
    In this section, we outline the details of Section~\ref{subsec:perturbation-theory}, which present various results in the direction of perturbation theory exploiting our $(k,v)$-Sobolev regularization property. We start with Proposition~\ref{prop:steady-state-perturbation-bound}:
    
    \begin{lem}[Duhamel's formula]\label{lem:perturbation-duhamel}
        Fix $v \in V$, $k > 0$, and $\varepsilon \in \mathbb R$. Let $T_t = e^{t\cL}$ and $S_t = e^{t\cB_\varepsilon}$ be $C_0$-semigroups on $\cT_1$ satisfying the \hyperref[assum:perturbation-common-domain]{common domain assumption}. Then, for every $t \ge 0$ and $x \in Y$,
        \begin{equation*}
            S_tx-T_tx=\varepsilon\int_0^tS_{t-s}\cE T_sx\,ds
            =\varepsilon\int_0^tT_{t-s}\cE S_sx\,ds.
        \end{equation*}
        If moreover $W_v^{k,1} \subset \cD(\cE)$ and $\cE:W_v^{k,1} \to \cT_1$ is bounded, the first identity extends to all trace-class inputs when $T_s$ satisfies the quantitative regularization bound \eqref{eq:quantitative-concentrated-regularization} with $\delta > k$. The second does so when $S_s$ satisfies that bound.
    \end{lem}
    \begin{proof}
        Strong continuity on $Y$ and boundedness of both generators from $Y$ to $\cT_1$ give continuity of both orbits in both generator graph norms. Thus \cite[Lem.~2.2.5]{Moebus.2025thesis} applies in both orientations. For general inputs, the semigroup property gives continuity of the integrands away from zero. Boundedness of $\cE:W_v^{k,1} \to \cT_1$ and the respective regularization estimate give an integrable majorant $C_t(s^{-k/\delta} + 1)\|x\|_1$ on $(0,t]$, since $k < \delta$. Density of $Y$ and dominated convergence yield the asserted Bochner integrals and identities.
    \end{proof}
    
    \begin{prop}[Perturbation bound for steady states]\label{appx-prop:steady-state-perturbation-bound}
        
        Fix $v \in V$, $k > 0$, and $\varepsilon \ge 0$. Let $(\cL,\cD(\cL))$ and $(\cB_\varepsilon,\cD(\cB_\varepsilon))$ generate QMSs satisfying the \hyperref[assum:perturbation-common-domain]{common domain assumption}. Assume $W_v^{k,1} \subset \cD(\cE)$ and, for some $\delta > k$, $\mu_k > 0$, and $c_k,c_1,c_2 \ge 0$,
        \begin{equation*}
            \|e^{t\cB_\varepsilon}\|_{1\to W_v^{k,1}}
            \le \left(\frac{k}{\delta\mu_k t}\right)^{k/\delta}
            + \left(\frac{c_k}{\mu_k}\right)^{k/(k+\delta)},\qquad t > 0,
        \end{equation*}
        and
        \begin{equation*}
            \|\cE(x)\|_1 \le c_1\|x\|_{W_v^{k,1}} + c_2\|x\|_1,
            \qquad x \in W_v^{k,1}.
        \end{equation*}
        Assume that $e^{t\cL}$ converges to a steady state $\overline\rho \in \cT_1$ with
        \begin{equation*}
            \|e^{t\cL} - \cP_{\overline\rho}\|_{1\to1}
            \le \widetilde C e^{-\gamma t},\qquad t \ge 1,\qquad
            \cP_{\overline\rho}(x) = \tr[x]\overline\rho,
        \end{equation*}
        for some $\widetilde C \ge 0$ and $\gamma > 0$, and that $\tr[\cE x] = 0$ for $x \in W_v^{k,1}$. The trace condition follows from the \hyperref[assum:perturbation-common-domain]{common domain assumption} and the relative Sobolev bound when $\varepsilon > 0$ and is unnecessary when $\varepsilon = 0$. Then, for every $t \ge 0$ and all states $\rho,\sigma \in \cT_1$,
        \begin{equation*}
            \|e^{t\cL}(\sigma) - e^{t\cB_\varepsilon}(\rho)\|_1
            \le
            \begin{cases}
                \|\sigma - \rho\|_1 + \varepsilon t^{1-k/\delta}\hat C_1\|\rho\|_1,&0 \le t \le 1,\\[4pt]
                \widetilde C e^{-\gamma t}\|\sigma - \rho\|_1
                + \varepsilon\hat C_{\rm long}\|\rho\|_1,&t \ge 1,
            \end{cases}
        \end{equation*}
        where
        \begin{equation*}
            \begin{gathered}
                a = k/\delta,\qquad A = c_1\left(\frac{k}{\delta\mu_k}\right)^a,
                \qquad B = c_1\left(\frac{c_k}{\mu_k}\right)^{k/(k+\delta)} + c_2,\\
                \hat C_1 = \frac{A}{1-a} + B,\qquad
                \hat C_{\rm long} = \hat C_1 + \widetilde C e^{-\gamma}
                \left[A\left(\gamma^{-1} + (1-a)^{-1}\right) + B\gamma^{-1}\right].
            \end{gathered}
        \end{equation*}
        The regularization parameters are those of the perturbed semigroup.
        
    \end{prop}
    \begin{proof}
        Similar to the proof of Theorem 6 in \cite{Szehr.2013}, we first apply the fundamental theorem of calculus:
        \begin{equation*}
            \|e^{t\cL}(\sigma)-e^{t(\cL+\varepsilon \cE)}(\rho)\|_1 \leq \|e^{t\cL}(\sigma-\rho)\|_1 + \varepsilon \int_0^t \|e^{s\cL}\cE\, e^{(t-s)(\cL+\varepsilon \cE)}(\rho)\|_1 \, ds .
        \end{equation*}
        Lemma~\ref{lem:perturbation-duhamel} justifies the Bochner integral. We choose $t^*=1$ and split into the cases $t\leq1$ and $t\geq1$. For the first case, we first use the contractivity of QMS, then the relative boundedness assumption, and finally the quantitative estimate \eqref{eq:quantitative-concentrated-regularization} for the perturbed semigroup:
        \begin{equation*}
            \begin{aligned}
                \|e^{t\cL}(\sigma)-e^{t(\cL+\varepsilon \cE)}(\rho)\|_1
                &\le \|\sigma-\rho\|_1 + \varepsilon \int_0^t \|\cE\, e^{(t-s)(\cL+\varepsilon \cE)}(\rho)\|_1 \, ds \\
                &\le \|\sigma-\rho\|_1 + \varepsilon \int_0^t \Bigl(c_1\|e^{(t-s)(\cL+\varepsilon \cE)}(\rho)\|_{W_v^{k,1}}
                + c_2\|\rho\|_1\Bigr) ds \\
                &\le \|\sigma-\rho\|_1 + \varepsilon \Bigl(\int_{0}^{t} c_1\left(\frac{k}{\delta \mu_k (t-s)}\right)^{\frac{k}{\delta}} ds + c_1 t \left(\frac{c_k}{\mu_k}\right)^{\frac{k}{k+\delta}}
                + c_2 t\Bigr) \|\rho\|_1 \\
                &\le \|\sigma-\rho\|_1 + \varepsilon \Bigl(c_1 \frac{\delta}{\delta-k} t^{1-\frac{k}{\delta}}\left(\frac{k}{\delta \mu_k}\right)^{\frac{k}{\delta}}
                + c_1 t \left(\frac{c_k}{\mu_k}\right)^{\frac{k}{k+\delta}}
                + c_2 t \Bigr)\|\rho\|_1 .
            \end{aligned}
        \end{equation*}

        Since $t\leq1$, the last line is at most $\|\sigma-\rho\|_1+\varepsilon t^{1-a}\hat C_1\|\rho\|_1$. For $t\geq1$, both $\sigma-\rho$ and $\cE e^{u\cB_\varepsilon}\rho$ are traceless. Thus the mixing assumption applies to these inputs, giving
        \begin{equation*}
            \begin{aligned}
                \|e^{t\cL}\sigma-e^{t\cB_\varepsilon}\rho\|_1
                &\leq\widetilde C e^{-\gamma t}\|\sigma-\rho\|_1\\
                &\quad+\varepsilon\left[\int_0^1(A(t-s)^{-a}+B)\,ds
                +\widetilde C\int_1^t e^{-\gamma s}(A(t-s)^{-a}+B)\,ds\right]\|\rho\|_1.
            \end{aligned}
        \end{equation*}
        The first integral is at most $\hat C_1=A/(1-a)+B$. For the second, split into $t-s<1$ and $t-s\geq1$. On the first part, integrate $(t-s)^{-a}$ and use $e^{-\gamma s}\leq e^{-\gamma}$. On the second, use $(t-s)^{-a}\leq1$ and integrate the exponential. Consequently,
        \begin{equation}\label{eq:mixing-convolution-bound}
            \int_1^t e^{-\gamma s}(A(t-s)^{-a}+B)\,ds
            \leq e^{-\gamma}\left[A\left((1-a)^{-1}+\gamma^{-1}\right)+B\gamma^{-1}\right].
        \end{equation}
        This gives $\hat C_{\rm long}$ and completes the proof. In this convolution the regularization time is $t-s$, not the propagation time $s$.
        
    \end{proof}

    Next, we prove Proposition~\ref{prop:invariant-subset-perturbation-bound} in detail.
    \begin{prop}[Perturbation bound for invariant subsets]\label{appx-prop:invariant-subset-perturbation-bound}
        
        Fix $v \in V$, $k > 0$, and $\varepsilon \ge 0$. Let $(\cL,\cD(\cL))$ and $(\cB_\varepsilon,\cD(\cB_\varepsilon))$ be generators of QMSs satisfying the \hyperref[assum:perturbation-common-domain]{common domain assumption}. Assume that the perturbed semigroup satisfies the quantitative bound \eqref{eq:quantitative-concentrated-regularization} with $\delta>k$, $\mu_k>0$, and $c_k\geq0$, and that
        \begin{equation*}
            \|\cE(\rho)\|_1\leq c_1\|\rho\|_{W_v^{k,1}}+c_2\|\rho\|_1
            \quad(\rho\in W_v^{k,1}).
        \end{equation*}
        If $e^{t\cL}$ converges to an invariant subset with contractive projection $\cP$ (with $\cP^\perp=\mathrm{id}-\cP$), i.e.,
        \begin{equation*}
            \|e^{t\cL}-\cP\|_{1\to1}\leq\widetilde C e^{-\gamma t},\quad t\geq1,
            \qquad \cP e^{t\cL}=e^{t\cL}\cP=\cP,
        \end{equation*}
        then for every $t \ge 0$ and any states $\rho,\sigma\in\cT_1$,
        \begin{equation*}
            \begin{aligned}
                \|\cP^\perp e^{t\cL}(\sigma)-\cP^\perp e^{t(\cL+\varepsilon\cE)}(\rho)\|_1
                \leq
                \begin{cases}
                    \|\cP^\perp(\sigma-\rho)\|_1+2\varepsilon t^{1-k/\delta}\hat C_1\|\rho\|_1,&t\leq1,\\[4pt]
                    \widetilde C e^{-\gamma t}\|\cP^\perp(\sigma-\rho)\|_1+
                    \varepsilon\hat C_{\rm long}^\perp\|\rho\|_1,&t\geq1.
                \end{cases}
            \end{aligned}
        \end{equation*}
        Here $a,A,B,\hat C_1$ are as in Proposition~\ref{prop:steady-state-perturbation-bound}, and
        \begin{equation*}
            \hat C_{\rm long}^{\perp}=2\hat C_1+\widetilde C e^{-\gamma}
            \left[A\left(\gamma^{-1}+(1-a)^{-1}\right)+B\gamma^{-1}\right].
        \end{equation*}
        The factor $2$ uses $\|\cP^\perp\|_{1\to1}\leq2$. This complementary projection need not be contractive.
        
    \end{prop}
    \begin{proof}
        For the case $t \leq 1$, we follow the proof of Proposition~\ref{prop:steady-state-perturbation-bound}, which shows
        
        \begin{equation*}
            \|\cP^\perp e^{t\cL}\sigma-\cP^\perp e^{t\cB_\varepsilon}\rho\|_1
            \leq\|\cP^\perp(\sigma-\rho)\|_1
            +2\varepsilon t^{1-a}\hat C_1\|\rho\|_1.
        \end{equation*}
        Here $\cP^\perp e^{t\cL}=e^{t\cL}\cP^\perp$ and $\|\cP^\perp\|\leq2$. For $t\geq1$, the initial difference is bounded by $\widetilde C e^{-\gamma t}\|\cP^\perp(\sigma-\rho)\|_1$. In Duhamel's integral, use
        \begin{equation*}
            \|\cP^\perp e^{s\cL}\|_{1\to1}
            \leq\begin{cases}2,&0\leq s\leq1,\\
                \widetilde C e^{-\gamma s},&s\geq1,
            \end{cases}
        \end{equation*}
        because $\cP^\perp e^{s\cL}=e^{s\cL}-\cP$. The preceding convolution estimate then gives $\hat C_{\rm long}^{\perp}$. No assumption that $\cP^\perp$ is contractive is needed.
        
        This completes the proof.
    \end{proof}

\section{Generation on concentrated Sobolev spaces}\label{appx:concentrated-generation}
    
    We prove Theorem~\ref{thm:concentrated-generation} by finite-Fock approximation. The weighted generation method follows \cite{Gondolf.2024}. Here we work directly with the concentrated weights and use the Lumer--Phillips and Trotter--Kato theorems \cite[Thms.~II.3.15 and III.4.8]{Engel.2000}.

    We write \eqref{eq:bosonic-GKLS} on $\cT_f$ as
    \begin{equation}\label{eq:polynomial-core-generator}
        \cL(x) = Gx + xG^\dagger + \sum_{a=1}^{J} L_a xL_a^\dagger,
        \qquad
        G = -iH - \frac12 \sum_{a=1}^{J} L_a^\dagger L_a,
    \end{equation}
    with formal adjoints on $\mathcal H_f$ as in Section~\ref{sec:prelim}.

    Write $m = |V|$ and $S = \1 + \sum_i N_i$. For $R \in \mathbb N$, define the joint Fock cutoff by
    \begin{equation*}
        P_R  = \sum_{\substack{\boldsymbol n \in \mathbb N_0^V \\ \sum_{i \in V} n_i \leq R-1}} \ketbra{\boldsymbol n}{\boldsymbol n},
        \qquad \Pi_R x = P_R x P_R.
    \end{equation*}
    Recall that $M_{v,0} = \1$ and $W_v^{0,1} = \cT_1$. Fix an integer $d \geq 1$ bounding the total degrees of $H,L_1,\ldots,L_J$.

    \begin{lem}[Weighted approximation]\label{appx-lem:concentrated-weight-calculus}
        For every $q \geq 0$ and $v \in V$, the sandwich $\cW_v^q$ is an isometric isomorphism onto $\cT_1$, with completely positive inverse $y \mapsto M_{v,q}^{-1/2} y M_{v,q}^{-1/2}$. The contractions $\Pi_R$ converge strongly to the identity in $W_v^{q,1}$, and the weights satisfy
        \begin{equation}\label{appx-eq:concentrated-total-comparison}
            \alpha_{v,\min} m^{-q} S^q \leq M_{v,q} \leq S^q,
            \qquad \alpha_{v,\min} = \min_i \alpha_{v,i} > 0.
        \end{equation}
        Moreover, the polynomial operator $(\cL,\cT_f)$ extends to a bounded map
        \begin{equation}\label{appx-eq:polynomial-weighted-map}
            \cL : W_v^{q+2d,1} \longrightarrow W_v^{q,1},
        \end{equation}
        and, for every $a > 0$ and $y \in W_v^{q+a,1}$,
        \begin{equation}\label{appx-eq:weighted-cutoff-tail}
            \|(\mathrm{id} - \Pi_R)y\|_{W_v^{q,1}} \leq C_{v,q,a} R^{-a/2} \|y\|_{W_v^{q+a,1}},
        \end{equation}
        where $C_{v,q,a}$ is independent of $R$.
    \end{lem}
    \begin{proof}
        The inverse sandwich gives the isometry and completeness. Commutation of $P_R$ with the weights reduces strong convergence to $P_R y P_R \to y$ in trace norm.

        To compare the weights, note that $S$ and $M_{v,q}$ are diagonal in the Fock basis $\ket{\boldsymbol n} = \bigotimes_{i \in V} \ket{n_i}$, with eigenvalues $1 + \sum_i n_i$ and $\sum_i \alpha_{v,i}(n_i + 1)^q$, respectively. The upper bound follows from $n_i + 1 \leq 1 + \sum_j n_j$ and $\sum_i \alpha_{v,i} = 1$. For the lower bound, choose a mode $j$ with maximal $n_j + 1$. Since $1 + \sum_i n_i \leq \sum_i(n_i + 1) \leq m(n_j + 1)$, its contribution alone gives
        \begin{equation*}
            \sum_i \alpha_{v,i}(n_i + 1)^q
            \geq \alpha_{v,j}(n_j + 1)^q
            \geq \alpha_{v,\min} m^{-q} \left(1 + \sum_i n_i\right)^q.
        \end{equation*}
        These scalar inequalities on every Fock vector give the operator comparison \eqref{appx-eq:concentrated-total-comparison}.

        For the polynomial estimate, we adapt the Fock-basis and factorization arguments of \cite[Lem.~E.2 and proof of Lem.~E.3]{Gondolf.2024}. A monomial of degree $D$ maps $\ket{\boldsymbol n}$ either to zero or to a multiple of $\ket{\boldsymbol n + \boldsymbol r}$, where $\boldsymbol r \in \mathbb Z^V$ is fixed. The coefficient is bounded by $C(1 + \sum_i n_i)^{D/2}$, and the total occupation changes by at most $D$. Thus the initial and final values of $1 + \sum_i n_i$ are comparable whenever the output is nonzero. Together with \eqref{appx-eq:concentrated-total-comparison}, this shows that, for every polynomial $B$ of degree at most $2d$, both
        \begin{equation*}
            M_{v,q}^{1/2} B M_{v,q+2d}^{-1/2}, \qquad
            M_{v,q+2d}^{-1/2} B M_{v,q}^{1/2}
        \end{equation*}
        extend to bounded operators. The second assertion also follows by applying the first to $B^\dagger$. For polynomials $B,C$ of degree at most $2d$, put $U_B=M_{v,q}^{1/2}BM_{v,q+2d}^{-1/2}$ and $V_C=M_{v,q+2d}^{-1/2}CM_{v,q}^{1/2}$. Factoring through the input sandwich gives
        \begin{equation*}
            \begin{aligned}
                \|BxC\|_{W_v^{q,1}}
                &=\|U_B(M_{v,q+2d}^{1/2}xM_{v,q+2d}^{1/2})V_C\|_1\\
                &\leq\|U_B\|_\infty\|x\|_{W_v^{q+2d,1}}\|V_C\|_\infty.
            \end{aligned}
        \end{equation*}
        Applying this bound to the terms of \eqref{eq:polynomial-core-generator} proves \eqref{appx-eq:polynomial-weighted-map}.

        For the cutoff estimate, put $A = M_{v,q}^{1/2} M_{v,q+a}^{-1/2}$. Monotonicity of the commuting weights gives $\|A\| \leq 1$, while \eqref{appx-eq:concentrated-total-comparison} gives $A^2=M_{v,q}M_{v,q+a}^{-1}\leq\alpha_{v,\min}^{-1}m^{q+a}S^{-a}$. Consequently,
        \begin{equation*}
            \|(\1 - P_R)A\| =\| (\1-P_R)A^2(\1-P_R)\|^{1/2} \leq \alpha_{v,\min}^{-1/2} m^{(q+a)/2} R^{-a/2}.
        \end{equation*}
        Apply $z-P_RzP_R=(\1-P_R)z+P_Rz(\1-P_R)$ to $z=M_{v,q+a}^{1/2}yM_{v,q+a}^{1/2}$. Since $A$ commutes with $P_R$,
        \begin{equation*}
            \begin{aligned}
                \|(\mathrm{id}-\Pi_R)y\|_{W_v^{q,1}}
                &=\|A(z-P_RzP_R)A\|_1\\
                &\leq2\|(\1-P_R)A\|\|A\|\|z\|_1\\
                &\leq2\alpha_{v,\min}^{-1/2}m^{(q+a)/2}R^{-a/2}
                \|y\|_{W_v^{q+a,1}},
            \end{aligned}
        \end{equation*}
        which proves \eqref{appx-eq:weighted-cutoff-tail}.

    \end{proof}

    \begin{proof}[Proof of Theorem~\ref{thm:concentrated-generation}]
        \emph{Cutoff estimates.} We first construct completely positive cutoff semigroups with the prescribed growth bounds to control the limiting semigroup. Define the bounded cutoff generator on $\cT_1$ by
        \begin{equation}\label{appx-eq:absorbing-cutoff-generator}
            G_R=P_RGP_R,\quad L_{a,R}=P_RL_aP_R,\quad
            \cL_Rx=G_Rx+xG_R^\dagger+\sum_aL_{a,R}xL_{a,R}^\dagger.
        \end{equation}
        By the bounded-generator characterization of completely positive semigroups \cite{Lindblad.1976}, $T_t^{(R)} = e^{t\cL_R}$ is completely positive. Moreover,
        \begin{equation}\label{appx-eq:absorbing-trace-defect}
            G_R+G_R^\dagger+\sum_aL_{a,R}^\dagger L_{a,R}
            =-\sum_aP_RL_a^\dagger(\1-P_R)L_aP_R\leq0
        \end{equation}
        so $T_t^{(R)}$ is trace nonincreasing and hence trace-norm contractive by \cite[Lem.~2.3.5]{Moebus.2025thesis}.

        Since $\cL^*(\1)=0$ when tested against vectors in $\mathcal H_f$, that is,
        \begin{equation*}
            \langle \psi,\cL^*(\1)\psi\rangle = 0,
            \qquad \psi \in \mathcal H_f,
        \end{equation*}
        order $k=0$ satisfies \eqref{eq:concentrated-core-generation-drift} with $\omega_{v,0}=0$. Fix $s \in \{0,k,l\}$ at the center $v$ and put $M = M_{v,s}$. Testing rank-one positive operators gives the form inequality $\cL^*(M)\leq\omega_{v,s}M$ on finite-Fock vectors. Since $P_R$ commutes with $M$, the following inequalities hold as quadratic forms:
        \begin{equation}\label{appx-eq:absorbing-weighted-drift}
            \cL_R^*(M)
            =P_R\cL^*(M)P_R
            -\sum_aP_RL_a^\dagger(\1-P_R)M(\1-P_R)L_aP_R
            \leq\omega_{v,s}P_RMP_R\leq\omega_{v,s}M.
        \end{equation}
        The conjugated generator $\cW_v^s \cL_R (\cW_v^s)^{-1}$ has a GKSL representation with bounded Hilbert-space operators $\widetilde G_R = M^{1/2}G_RM^{-1/2}$ and $\widetilde L_{a,R} = M^{1/2}L_{a,R}M^{-1/2}$, which satisfy
        \begin{equation}
            \widetilde G_R+\widetilde G_R^\dagger+\sum_a\widetilde L_{a,R}^\dagger\widetilde L_{a,R}\leq\omega_{v,s}\1.
        \end{equation}
        Subtracting $\omega_{v,s}\mathrm{id}$ from this conjugated generator therefore gives a completely positive, trace-nonincreasing semigroup. Again by \cite[Lem.~2.3.5]{Moebus.2025thesis}, its norm on the full trace class is at most one. Hence
        \begin{equation}\label{appx-eq:uniform-cutoff-semigroup}
            \|T_t^{(R)}\|_{W_v^{s,1}\to W_v^{s,1}}\leq e^{\omega_{v,s}t}.
        \end{equation}

        \emph{Generation at the target order and in trace class.} We next verify the Lumer--Phillips hypotheses at orders $0$ and $k$, using the upper-order bound to obtain dense range and hence the asserted generation and core properties. Fix $q\in\{0,k\}$ and work in $X=W_v^{q,1}$, with $\omega_{v,0}=0$. For every $x\in\cT_f$, polynomial coefficients imply $\cL_Rx=\cL x$ for all sufficiently large $R$. Together with \eqref{appx-eq:uniform-cutoff-semigroup}, this shows that $\cL - \omega_{v,q}\mathrm{id}$ is dissipative on $\cT_f$, meaning that
        \begin{equation*}
            \|((\lambda + \omega_{v,q})\mathrm{id} - \cL)x\|_X \geq \lambda \|x\|_X,
            \qquad \lambda > 0,\quad x \in \cT_f.
        \end{equation*}
        The domain $\cT_f$ is dense in $X$ by Lemma~\ref{appx-lem:concentrated-weight-calculus}. We next prove that $(\cL,\cT_f)$ is closable in $X$. Suppose that $x_j\to0$ and $\cL x_j\to y$ in $X$. Since $X$ embeds continuously into $\cT_1$, both convergences hold in trace norm. For each fixed pair of Fock multi-indices, the corresponding matrix entry of $\cL x_j$ is a finite linear combination of fixed matrix entries of $x_j$. It therefore converges to zero. Trace-norm convergence of $\cL x_j$ to $y$ also implies entrywise convergence, so every matrix entry of $y$ vanishes and hence $y=0$.

        For the auxiliary order $l$ in the hypothesis, set $a=l-q-2d>0$ and fix $\lambda>\max\{\omega_{v,q},\omega_{v,l}\}$. For $\xi\in\cT_f$ and $R$ with $\Pi_R\xi=\xi$, put $x_R=(\lambda-\cL_R)^{-1}\xi$. By \cite[Thm.~II.1.10]{Engel.2000},
        \begin{equation}
            x_R=\int_0^\infty e^{-\lambda t}T_t^{(R)}\xi\,dt,
            \qquad
            \|x_R\|_{W_v^{l,1}}\leq\frac{\|\xi\|_{W_v^{l,1}}}{\lambda-\omega_{v,l}}.
        \end{equation}
        The finite-dimensional range of $\Pi_R$ is invariant under $\cL_R$, so $x_R\in\cT_f$, $\Pi_Rx_R=x_R$, and $\cL_Rx_R=\Pi_R\cL x_R$. Lemma~\ref{appx-lem:concentrated-weight-calculus} yields
        \begin{equation}\label{appx-eq:concentrated-dense-range}
            \begin{aligned}
                \|(\lambda-\cL)x_R-\xi\|_{W_v^{q,1}}
                &=\|(\lambda-\cL)x_R-(\lambda-\cL_R)x_R\|_{W_v^{q,1}}\\
                &=\|(\mathrm{id}-\Pi_R)\cL x_R\|_{W_v^{q,1}}\\
                &\leq C_{v,q,a} R^{-a/2} \|\cL x_R\|_{W_v^{q+a,1}}\\
                &\leq C_{v,q,l}R^{-a/2}\|x_R\|_{W_v^{l,1}}
                \longrightarrow0,
            \end{aligned}
        \end{equation}
        where $C_{v,q,l}$ is independent of $R$. Since $\xi\in\cT_f$ was arbitrary, \eqref{appx-eq:concentrated-dense-range} shows that $\cT_f$ is contained in the closure of $(\lambda-\cL)\cT_f$ in $X$. As $\cT_f$ is dense in $X$, the range $(\lambda-\cL)\cT_f$ is dense in $X$. By the Lumer--Phillips theorem, the closure $\cL_{v,q}$ of $(\cL,\cT_f)$ generates a $C_0$-semigroup $T_t^{v,q}$ with norm at most $e^{\omega_{v,q}t}$. This step uses only the cutoff bound at $l$, not generation at $l$.

        \emph{Consistency and trace preservation.} We now show that the semigroups constructed on $\cT_1$ and $W_v^{k,1}$ agree, and that the trace-class semigroup is a QMS. For every $x\in\cT_f$, the equality $\cL_Rx=\cL x$ holds for all sufficiently large $R$. By the definition of $\cL_{v,q}$ as the closure of $(\cL,\cT_f)$, the space $\cT_f$ is a core for $\cL_{v,q}$. The Trotter--Kato approximation theorem \cite[Thm.~III.4.8]{Engel.2000} therefore gives $T_t^{(R)}x\to T_t^{v,q}x$ in $W_v^{q,1}$, locally uniformly in $t$. At $q=0$, denote the limit by $T_t$ and its generator by $\overline{\cL}$. For $x\in W_v^{q,1}$, the continuous embedding $W_v^{q,1}\hookrightarrow\cT_1$ implies that $T_t^{(R)}x\to T_t^{v,q}x$ also in trace norm. The trace-class construction gives $T_t^{(R)}x\to T_tx$ in the same norm, so uniqueness of the limit yields $T_t^{v,q}x=T_tx$. Thus $T_t^{v,q}$ is the restriction of $T_t$ to $W_v^{q,1}$. Complete positivity passes to $T_t$, since at every finite matrix level the positive cone is closed in trace norm. Since $\tr(\cL x)=0$ on $\cT_f$ and $\cT_f$ is a core for $\overline{\cL}$, $\tr(\overline{\cL}x)=0$ for every $x\in\cD(\overline{\cL})$. Differentiating $\tr(T_tx)$ on this domain, then using density, proves trace preservation on $\cT_1$.

        \emph{Admissibility at the upper order.} We use trace-class convergence and the cutoff moment bound to establish admissibility and the stated bound at $l$, without a generator-core assertion at this order. Fix $v$ and write $M=M_{v,l}$. The map $\rho\mapsto\tr(M\rho)$ on positive trace-class operators is lower semicontinuous in trace norm, since it is the supremum of the continuous functionals $\rho\mapsto\tr(P_RMP_R\rho)$. Hence trace-norm convergence of the cutoff semigroups gives
        \begin{equation*}
            \tr(MT_t\rho)
            \leq\liminf_{R\to\infty}\tr(MT_t^{(R)}\rho)
            \leq e^{\omega_{v,l}t}\tr(M\rho),
            \qquad 0\leq\rho\in W_v^{l,1}.
        \end{equation*}
        For $x\in W_v^{l,1}$, decompose the real and imaginary parts of $M^{1/2}xM^{1/2}$ into positive parts and conjugate back by $M^{-1/2}$. The positive moment bound then gives invariance of $W_v^{l,1}$. The $2\times2$ argument in \cite[proof of Thm.~4.1.1]{Moebus.2025thesis}, with input and output weight $M$, yields \eqref{eq:concentrated-generation-bound} at $s = l$ on the full space.

        For $0 \leq \rho \in \cT_f$, set $z_t = M^{1/2}T_t(\rho)M^{1/2}$ and choose $P = P_{R_0}$ with $P\rho P = \rho$. Since $P$ commutes with $M$, we have $Pz_0P = z_0$, and trace-class strong continuity gives $Pz_tP \to z_0$ in trace norm as $t \downarrow 0$. The moment bound then yields
        \begin{equation*}
            0 \leq \tr[(\1 - P)z_t]
            \leq e^{\omega_{v,l}t}\tr[z_0] - \tr[Pz_tP]
            \longrightarrow 0.
        \end{equation*}
        For every positive trace-class $z$,
        \begin{equation*}
            \|z-PzP\|_1
            \leq2\sqrt{\tr(z)\tr[(\1-P)z]}+\tr[(\1-P)z],
        \end{equation*}
        by the block decomposition and the Hilbert--Schmidt Cauchy--Schwarz inequality. Hence $\|z_t-Pz_tP\|_1\to0$. Together with $Pz_tP\to z_0$ in trace norm, this gives
        \begin{equation*}
            \begin{aligned}
                \|T_t\rho-\rho\|_{W_v^{l,1}}
                &=\|z_t-z_0\|_1\\
                &\leq\|z_t-Pz_tP\|_1+\|Pz_tP-z_0\|_1
                \longrightarrow0.
            \end{aligned}
        \end{equation*}
        The linear span of the positive elements of $\cT_f$ is dense in $W_v^{l,1}$. The local uniform norm bound extends strong continuity to every element of that space, proving admissibility.

        \emph{Maximum norm.} When stability holds at both orders for every center, we combine the fixed-center conclusions to obtain admissibility, the maximum-norm bounds, and the core property at $k$. Each centered construction has the same trace-class generator, namely the closure of $(\cL,\cT_f)$, and therefore gives the same trace-class semigroup $(T_t)_{t\geq0}$. For a fixed finite $V$, the finite maximum of the centered norms gives strong continuity on $W_{\max}^{s,1}$ for $s\in\{k,l\}$. At order $k$, comparison of the commuting weights in \eqref{appx-eq:concentrated-total-comparison} makes the centered norms, and hence their generator graph norms, equivalent, so $\cT_f$ remains a generator core for the maximum norm. Taking the maximum in \eqref{eq:concentrated-generation-bound} proves \eqref{eq:concentrated-max-generation-bound}. The volume-dependent comparison constants do not enter the semigroup bounds, which retain the prescribed $\omega_{v,k}$.
    \end{proof}

    \begin{lem}[Generator domains and transfer of the concentrated Sobolev regularization condition]\label{lem:core-drift-transfer}
        Let $(\cL,\cT_f)$ satisfy Assumption~\ref{assum:concentrated-polynomial}, and suppose its trace-class closure generates a QMS $(T_t)_{t\geq0}$. Fix $v\in V$ and $q\geq0$. If $W_v^{q,1}$ is admissible and its generator $\cL_{v,q}$ is the closure of $(\cL,\cT_f)$ in that space, then, for every $r\geq q+2d$,
        \begin{equation}\label{appx-eq:high-weight-generator-domain}
            W_v^{r,1}\subset\mathcal D(\cL_{v,q}),\qquad
            \|\cL_{v,q}x\|_{W_v^{q,1}}\leq C_{v,q,r}\|x\|_{W_v^{r,1}}.
        \end{equation}
        Here $q=0$ denotes the ambient generator. Suppose this generator-core property holds at $q=k>0$ and that the concentrated Sobolev regularization condition holds for $k,\delta,\mu>0$ and $c\geq0$:
        \begin{equation}\label{appx-eq:strong-core-drift}
            \tr(M_{v,k}\cL\rho)
            \leq-\mu\tr(M_{v,k+\delta}\rho)+c\tr\rho,
            \qquad \rho\in\cT_f,\quad\rho\geq0.
        \end{equation}
        If $W_v^{l,1}$ is admissible for some $l\geq\max\{k+2d,k+\delta\}$, then, for every positive $\rho\in W_v^{l,1}$, the orbit is continuously differentiable in $W_v^{k,1}$ and
        \begin{equation}\label{appx-eq:generated-moment-drift}
            \frac{d}{dt}\tr(M_{v,k}T_t\rho)
            \leq-\mu\tr(M_{v,k+\delta}T_t\rho)+c\tr\rho.
        \end{equation}
        In particular this holds for positive finite-Fock-support initial operators.
    \end{lem}
    \begin{proof}
        By \eqref{appx-eq:polynomial-weighted-map} and the contractive embedding $W_v^{r,1} \hookrightarrow W_v^{q+2d,1}$, the map $\cL : W_v^{r,1} \to W_v^{q,1}$ is bounded. The cutoffs $\Pi_R x \to x$ converge in the source space, and their polynomial images converge in the target space. Closedness of the generated graph closure proves \eqref{appx-eq:high-weight-generator-domain}.

        On $W_v^{l,1}$ the functionals $x\mapsto\tr(M_{v,s}x)$ are continuous for $0\leq s\leq l$, and $x\mapsto\tr(M_{v,k}\cL_{v,k}x)$ is continuous by \eqref{appx-eq:high-weight-generator-domain}. Positive Fock cutoffs therefore extend \eqref{appx-eq:strong-core-drift} to every positive $x\in W_v^{l,1}$. Admissibility at order $l$ gives a continuous positive orbit in that space. The domain inclusion places it in $\cD(\cL_{v,k})$, so \cite[Lem.~II.1.3(ii)]{Engel.2000} gives continuous differentiability in $W_v^{k,1}$. Applying the continuous moment functional and trace preservation proves \eqref{appx-eq:generated-moment-drift}, with a right derivative at zero. No generator-core property at $l$ is needed.
    \end{proof}

\end{document}